\documentclass[11pt]{article}

\usepackage{graphicx,amsmath,amsbsy,amssymb,amsthm,algorithmic,algorithm,url,subcaption}
\usepackage[usenames]{color}
\usepackage{cite}
\usepackage[colorlinks = false]{hyperref}

\newtheorem{thm}{Theorem}[section] 
\newtheorem{lemma}{Lemma}[section] 

\newtheorem{prop}{Proposition}[section]

\newcommand{\sign}{\mathrm{sign}}
\newcommand{\Grid}{\mathrm{Grid}}
\newcommand{\near}{\mathrm{near}}

\newcommand{\far}{\mathrm{far}}

\newcommand{\diag}{\mathrm{diag}}
\newcommand{\Toep}{\mathrm{Toep}}
\newcommand{\trace}{\mathrm{trace}}

\newcommand{\real}{\mathbb{R}}
\newcommand{\complex}{\mathbb{C}}

\newcommand{\cF}{\mathcal{F}}
\newcommand{\cA}{\mathcal{A}}

\newcommand{\cB}{\mathcal{B}}

\newcommand{\cE}{\mathcal{E}}

\newcommand{\bB}{{\boldsymbol{B}}}
\newcommand{\bp}{{\boldsymbol{p}}}
\newcommand{\bQ}{{\boldsymbol{Q}}}
\newcommand{\bD}{\boldsymbol{D}}
\newcommand{\bW}{\boldsymbol{W}}
\newcommand{\bE}{\boldsymbol{E}}

\newcommand{\bU}{\boldsymbol{U}}

\newcommand{\bT}{\boldsymbol{T}}

\newcommand{\bL}{\boldsymbol{L}}

\newcommand{\bx}{\boldsymbol{x}}

\newcommand{\balpha}{\boldsymbol{\alpha}}
\newcommand{\bbeta}{\boldsymbol{\beta}}
\newcommand{\bSigma}{\boldsymbol{\Sigma}}
\newcommand{\bA}{\boldsymbol{A}}
\newcommand{\by}{\boldsymbol{y}}
\newcommand{\bY}{\boldsymbol{Y}}
\newcommand{\bX}{\boldsymbol{X}}
\newcommand{\bK}{\boldsymbol{K}}
\newcommand{\bR}{\boldsymbol{R}}

\newcommand{\bZ}{\boldsymbol{Z}}
\newcommand{\bV}{\boldsymbol{V}}
\newcommand{\be}{\boldsymbol{e}}
\newcommand{\ba}{\boldsymbol{a}}
\newcommand{\bb}{\boldsymbol{b}}

\newcommand{\bg}{\boldsymbol{g}}
\newcommand{\bh}{\boldsymbol{h}}
\newcommand{\bq}{\boldsymbol{q}}

\newcommand{\bz}{\boldsymbol{z}}

\newcommand{\bv}{\boldsymbol{v}}
\newcommand{\bu}{\boldsymbol{u}}

\newcommand{\blambda}{\boldsymbol{\lambda}}
\newcommand{\bI}{\boldsymbol{I}}
\newcommand{\bzero}{\boldsymbol{0}}

\newcommand{\norm}[2][2]{\left\| #2 \right\|_{#1}}

\newcommand{\expval}{\mathbb{E}}

\begin{document}
\title{{Super-Resolution of Complex Exponentials from\\ Modulations with Unknown Waveforms}}

\author{Dehui~Yang, Gongguo~Tang, and~Michael~B.~Wakin\thanks{The authors are with the
Department of Electrical Engineering and Computer Science,
Colorado School of Mines, Golden, CO 80401 USA (e-mail: \url{dyang@mines.edu}; \url{gtang@mines.edu}; \url{mwakin@mines.edu}). D. Yang and M. B. Wakin were partially supported by NSF CAREER grant CCF-1149225 and NSF grants CCF-1409258, CCF-1409261. G. Tang was supported by NSF grant CCF-1464205.}
\thanks{Parts of the results in this paper were presented at the 41st IEEE International
Conference on Acoustics, Speech, and Signal Processing (ICASSP), Shanghai, China, March 2016 \cite{yang2016non}.}
}
\date{January 2016; Last Revised August 2016}

\maketitle

\begin{abstract}
Super-resolution is generally referred to as the task of recovering fine details from coarse information. Motivated by applications such as single-molecule imaging, radar imaging, etc., we consider parameter estimation of complex exponentials from their modulations with unknown waveforms,
allowing for non-stationary blind super-resolution. This problem, however, is ill-posed since both the parameters associated with the complex exponentials and the modulating waveforms are unknown. To alleviate this, we assume that the unknown waveforms live in a common low-dimensional subspace. Using a lifting trick, we recast the blind super-resolution problem as a structured low-rank matrix recovery problem. Atomic norm minimization is then used to enforce the structured low-rankness, and is reformulated as a semidefinite program that is solvable in polynomial time. We show that, up to scaling ambiguities, exact recovery of both of the complex exponential parameters and
the unknown waveforms is possible when the waveform subspace is random and the number of measurements is proportional to the number of degrees of freedom in the problem. 
Numerical simulations support our theoretical findings, showing that non-stationary blind super-resolution using atomic norm minimization is possible.
\end{abstract}


\section{Introduction}
\subsection{Motivation}
Super-resolution refers to techniques for enhancing the resolution of imaging systems. It finds applications in a variety of practical problems, including single-molecule microscopy, computational photography, astronomy, radar imaging. For example, in single-molecule imaging \cite{huang2008three,yildiz2003myosin}, one is interested in studying the individual behavior of molecules from measurements of an ensemble of molecules. The measurements, however, only contain the average characteristics of the molecules with fine details smeared out by the point spread function of the imaging process. Super-resolution aims to recover these fine details by localizing individual molecules, and consequently, enhance the performance of the imaging system.

In this paper, we consider super-resolution of unknown complex exponentials from their modulations with unknown waveforms. This extends super-resolution to the \textit{blind} and \textit{non-stationary}\footnote{Our choice of the term ``non-stationary'' is inspired by its use in non-stationary deconvolution~\cite{margrave2011gabor}.} scenario. More specifically, consider the observation model
\begin{equation}
\by (n) = \sum_{j = 1}^J c_j e^{-i2\pi n \tau_j} \bg_j(n),
\label{model}
\end{equation}
where $\left\{ \by (n) \in \mathbb{C}\right\}$ are samples of a continuous-time output, $\left\{ c_j\right\} \subset \mathbb{C}$ and $\left\{ \tau_j \right\} \subset [0,1)$ are unknown coefficients and parameters associated with the complex exponentials, and $\{\bg_j(n)\}$ are samples of unknown waveforms, whose forms vary with the index $j$. Our goal is to recover $\{\tau_j \}$ and $\{ c_j \}$, as well as the samples of the unknown waveforms $\{ \bg_j(n)\}$. To make this otherwise ill-posed problem well-posed, we assume that the unknown waveforms $\left\{ \bg_j\right\}$ belong to a common and known low-dimensional subspace.

Model (\ref{model}) encompasses a wide spectrum of applications. Here we list three stylized examples that can be modeled using our general mathematical framework. \\
{\bf Super-resolution with unknown point spread functions}: In applications such as single-molecule microscopy, one is interested in super-resolving and localizing unknown point sources from their convolutions with point spread functions. Quite often, the point spread function, however, cannot be perfectly known. The point spread function may also depend on the locations of the point sources. This is the case in 3D single-molecule microscopy \cite{quirin2012optimal}, where the point spread function depends on the depth ($z$-axis) of the target, demanding a super-resolution technique that handles unknown and space-varying system functions. Another example is the non-stationary blind deconvolution of seismic data~\cite{margrave2011gabor}. Here the goal is to retrieve the time domain reflectivity of the earth from its convolution with (non-stationarily) attenuated seismic waves from samples of the seismic trace. Yet other non-stationary blind super-resolution applications include computational photography \cite{fergus2006removing} and astronomy \cite{starck2002deconvolution}. Finally, one further application involving simultaneous super-resolution and calibration of unknown waveforms is the blind multi-path channel identification problem in multi-user communication systems \cite{luo2006low}. At the receiver, one must estimate the multi-path delays of unknown waveforms set by different users. For all of these applications, the goal is to determine the unknown delays $\{ \tau_j \}$ and coefficients $\{c_j\}$ from observations of the form
\begin{equation}
\label{eq:exsuper}
y(t) = \sum_{j =1}^J c_j  g_j(t - \tau_j)
\end{equation}
with $\{g_j(t)\}$ being the unknown point spread functions. By taking Fourier transform on both sides of (\ref{eq:exsuper}), we obtain
\begin{equation}
\label{eq:exsuperfourier}
\widehat{y}(f) = \sum_{j = 1}^J c_j e^{-i2\pi f\tau_j} \widehat{g}_j (f),
\end{equation}
which takes the form of (\ref{model}) when sampled. The goal is to simultaneously recover $\{c_j\}, \{\tau_j\}$ and samples of the point spread functions $\{\widehat{g}_j(f)\}$.
\vspace{0.05in}\\
{\bf Parameter estimation in radar imaging}: In radar imaging \cite{heckel2014super}, one is concerned with estimating the distances and velocities of the targets relative to the radar. These quantities can be inferred by estimating the unknown delay-Doppler parameters $(\mu_j, \nu_j)$, from the following signal model:
\begin{equation}
\label{eq:exradar}
y(t) = \sum_{j = 1}^J c_j e^{i2\pi \nu_j t}x(t - \mu_j),
\end{equation}
where both the transmitted waveform $x(t)$ and the received waveform $y(t)$ are known. We note that $\nu_j$ and $\mu_j$ can be arbitrary and do not necessarily lie on a grid. It is easy to see that sampling (\ref{eq:exradar}) also produces (\ref{model}).
\vspace{0.05in}\\
{\bf Frequency estimation with damping}: In applications such as nuclear magnetic resonance spectroscopy \cite{cai2015robust}, the signal is the superposition of complex exponentials with unknown frequencies $\{f_j\}$ and damping factors $\{ \varsigma_j \}$:
\begin{equation}
\label{eq:exdamping}
y(t) = \sum_{j = 1}^J c_j e^{i2\pi f_j t} e^{-\varsigma_j t}.
\end{equation}
By sampling the continuous variable $t$ in (\ref{eq:exdamping}), we again obtain an instance of (\ref{model}). Here the modulating waveforms $\bg_j(n)$ are samples of the damping terms $e^{-\varsigma_j t}$.

In some cases, to help regularize the inverse problems above, it may be appropriate to assume, as we do, that the unknown waveforms $\left\{ \bg_j\right\}$ belong to a known low-dimensional subspace. In super-resolution imaging, for example, point spread functions can often be modeled as Gaussians; see~\cite{huang2008three,quirin2012optimal} and references therein. When the widths of the point spread functions are unknown, however, a dictionary can be constructed consisting of Gaussian functions with different variances. Applying principal component analysis (PCA) on the constructed dictionary reveals an approximate low-dimensional subspace structure that captures the unknown point spread functions. We demonstrate this in our numerical experiments in Section~\ref{sec:Numericalsimulations}. Further, in multi-user communication systems it may be reasonable to assume that the unknown waveforms transmitted by different users belong to a subspace; in addition, the multiple received copies of a single user's waveform will all be identical (save for a delay, which becomes part of the modulation term in~\eqref{model}). On the other hand, in radar imaging, the subspace spanned by sampled, shifted copies of the transmitted waveform may not always have a low dimension. Related works such as~\cite{heckel2014super} may give sharper guarantees in this case.

\subsection{Related Work}
In the past few years, super-resolution via convex programming has become a popular approach since convex methods usually come with strong theoretical guarantees and robustness to noise and outliers. In \cite{candes2014towards}, a general mathematical framework for super-resolution using total variation (TV) norm minimization is proposed. The goal there is to super-resolve the unknown locations in $[0,1)$ of point sources from low-frequency samples of the spectrum. This TV norm minimization problem can be recast as a computationally efficient semidefinite program (SDP) \cite{bhaskar2013atomic}. It is shown that one can super-resolve $J$ point sources from $O(J)$ samples under a minimum separation condition. We note that this approach, however, requires perfect information of the point spread function. Based on \cite{candes2014towards}, \cite{candes2013super,tang2015near} study the robustness of TV norm minimization for super-resolution by considering the noisy data case; \cite{morgenshtern2015super} extends the super-resolution problem to the case when the point sources are positive; \cite{duval2015exact} examines the recovery property of sparse spikes using TV norm minimization through studying the non-degeneracy of the dual certificate. Recent work  \cite{schiebinger2015superresolution} studies the super-resolution problem without separation. In \cite{fernandez2015super}, the author considers super-resolution for demixing and super-resolution of multiple signals with a common support. In \cite{tang2013compressed}, driven by applications in line spectral estimation, an atomic norm minimization scheme is proposed for super-resolution of arbitrary unknown frequencies from random time samples of a superposition of complex exponentials. It has been shown that the sample complexity is proportional to the number of frequencies (up to a polylogarithmic factor) for exact frequency estimation. It is also worth mentioning subsequent work based on \cite{candes2014towards,tang2013compressed}. In \cite{li2014off,yang2014exact}, the authors study the problem of frequency estimation when multiple measurement vectors (MMV) are available. \cite{chen2014robust} proposes an enhanced matrix completion algorithm for frequency estimation from limited time samples by converting spectral sparsity in the model into a low-rank structure of the block Hankel matrix.

Another line of related work addresses the blind deconvolution problem. In \cite{ahmed2014blind}, the bilinear blind deconvolution problem is reformulated as a rank-one matrix sensing problem. A nuclear norm minimization program is then utilized for rank-one matrix recovery. It is shown that by employing subspace models for both signals, one can recover two length-$L$ vectors from their circular convolution when $L = O(Q+K)$, where $Q$ and $K$ are the dimensions of the two subspaces. Following this general idea of lifting for blind deconvolution using convex programming, \cite{aliahmedcov} considers the problem of blind deconvolution when multiple unknown inputs belong to a known and diverse subspace; \cite{ling2015blind} extends the work in \cite{ahmed2014blind} from rank-one case to a general rank-$r$ matrix sensing problem, achieving simultaneous blind deconvolution and demixing.
In \cite{lee2013near}, the authors propose an alternating minimization scheme for blind deconvolution under a sparsity model for the underlying signals. In \cite{tang2014convex}, the authors propose a nuclear norm minimization algorithm for blind deconvolution using random masks. More recently, \cite{bahmani2014lifting} generalizes the problem studied in \cite{tang2014convex} by considering the effect of subsampling in the measurement process. Other related works along this line include  \cite{li2015identifiability,choudhary2014identifiability}, which study conditions for the uniqueness of blind deconvolution.

Our work is most closely related to the recent works \cite{chi2015guaranteed,ling2015self}. In \cite{ling2015self}, a biconvex problem for simultaneous sparse recovery and unknown gain calibration is studied. In their work, a subspace model is employed for the unknown gains to make the problem well-posed. It is worth mentioning that they use $\ell_1$ minimization as a convex program, which is different from ours. Then, a sample complexity bound that is suboptimal is derived for sparse recovery and self-calibration. Inspired by \cite{ling2015self}, \cite{chi2015guaranteed} considers a super-resolution problem that has a similar setup to \cite{candes2014towards}, except that the point spread function is assumed unknown. By employing a subspace model for the point spread function, an atomic norm minimization program is formulated for simultaneous super-resolution of point sources and recovery of the unknown point spread function. The atomic norm minimization problem therein is recast as an SDP. The sample complexity bound derived there, however, is suboptimal.
As we explain in Section~\ref{subsec:maincontribution}, our work further generalizes the model in \cite{chi2015guaranteed} to the non-stationary case, where the point spread functions can vary with the point sources.

In \cite{heckel2014super}, super-resolution radar is formulated as a convex optimization program. In particular, the signal is modeled as a superposition of delayed and Doppler shifted versions of the template waveform, which is the same as model (\ref{eq:exradar}).
It should be pointed out that our model (\ref{model}) can be utilized for this problem as well.
Therefore, our proposed blind super-resolution method can conceivably be used for super-resolution radar.

Lastly, we would like to mention that the signal model in our work has both low-rank and spectrally sparse structures, and thus is simultaneously structured. Consistent with~\cite{oymak2015simultaneously}, we can achieve the information-theoretic limit on the measurement bound (up to a polylogarithmic factor) not by a combination of convex objectives but rather through a single convex objective---in this case via atomic norm minimization.

\subsection{Main Contributions}
\label{subsec:maincontribution}
Our contributions are twofold. First, we propose a general model for non-stationary blind super-resolution, which arises in a variety of disciplines. Our non-stationary blind super-resolution problem is naturally non-convex. By utilizing a subspace model for the unknown waveforms and a lifting trick \cite{ahmed2014blind,candes2015phase}, we relax the non-stationary blind super-resolution problem using atomic norm minimization, which can be further formulated as an SDP. Second, we derive a sample complexity bound that is near information-theoretically optimal under assumptions on the minimum separation of the $\tau_j$'s and on the randomness and incoherence properties of the subspace. Specifically, assuming that the subspace has dimension $K$, we show that when the number of measurements is proportional to the number of degrees of freedom in the problem, i.e., $O(JK)$ (up to a polylogarithmic factor), the non-stationary blind super-resolution problem is solvable by an SDP. Furthermore, we can faithfully recover $\{ \tau_j \}$, $\{ c_j \}$, and the samples of the unknown waveforms $\{ \bg_j(n) \}$.

It is also worth mentioning the recent work \cite{chi2015guaranteed}, which can be viewed as a special case of our general non-stationary blind super-resolution framework by assuming all of the unknown waveforms are the same. Our model is more realistic and powerful due to its generality. As illustrated by the examples in the introduction, our framework also covers a wider range of non-convex inverse problems beyond the super-resolution problem with unknown point spread functions, including blind multi-path channel identification in communication systems, parameter estimation in radar imaging and frequency estimation with damping. Additionally, on the theoretical side, we improve the sample complexity bound in \cite{chi2015guaranteed} from $O(J^2K^2)$ to $O(JK)$, up to a polylogarithmic factor. We elaborate on comparisons with~\cite{chi2015guaranteed} in Section~\ref{sec:discussmainresult}.

\subsection{Notation and Organization of the Paper}
Throughout the paper, the following notation is adopted. We use boldface letters $\bX, \bY$ and $\bx, \by$ to denote matrices and vectors, respectively.  For a vector $\bv$, $\norm[2]{\bv}$ is used to denote the $\ell_2$ norm of $\bv$. For a matrix $\bX$, $\norm[]{\bX}$ and $\norm[F]{\bX}$ represent the operator norm and the Frobenius norm of the matrix $\bX$, respectively. An $M\times N$ zero matrix is denoted as $\bzero_{M\times N}$. We also use $\bzero_M$ and $\bI_M$ to denote an $M \times M$ zero matrix and an $M\times M$ identity matrix, respectively. We use the matrix inequality notation $\bX \preceq \bY$ to represent that $\bY-\bX$ is positive semidefinite. Conventional notations $\langle \cdot, \cdot \rangle$, $\trace(\cdot)$, $(\cdot)^H$, $(\cdot)^T$ and $(\cdot)^{*}$ are used to denote the inner product, trace, Hermitian, transpose, and conjugate operations, respectively. For a set $\Omega$, we use $|\Omega|$ to denote its cardinality. $\expval [\cdot]$ and $\mathbb{P}\left\{\cdot \right\}$ denote expectation and probability of the underlying event. We use a calligraphic letter $\cB$ to denote a linear operator.

The rest of the paper is organized as follows. Section \ref{sec: problemformulation}
introduces our problem setup, its connection to the atomic norm minimization framework via a lifting trick, and an exact SDP reformulation of the atomic norm minimization. In Section \ref{sec:recovery guarantees}, we present the main theorem and discuss its implications. Section \ref{sec:Numericalsimulations} provides some numerical simulations to support and illustrate our theoretical findings. The detailed proof of our main theorem is presented in Section \ref{sec:detailedproofofmainthm}.
Finally, concluding remarks are given in Section \ref{sec:conclusions}.

\section{Problem Formulation}
\label{sec: problemformulation}
\subsection{Problem Setup via Atomic Norm Minimization}
Consider the model
\begin{equation}
\label{eq:model} \by (n) = \sum_{j = 1}^J c_j e^{-i2\pi n \tau_j} \bg_j(n), ~~n = -2M, \dots, 2M,
\end{equation}
where $\by(n) \in \mathbb{C}$ are observations, $(c_j, \tau_j)$ are unknown parameters of complex exponentials $c_j e^{-i2\pi n \tau_j}$, and $\bg_j(n)$ are samples of unknown waveforms. Without loss of generality, we assume that $\tau_j \in [0,1), j = 1,\dots, J$. Our goal is to recover $\tau_j$, $c_j$, and $\bg_j(n)$ from the samples $\by(n)$. It is apparent that one can only recover $c_j$ and $\bg_j(n)$ up to a scaling factor due to the multiplicative form in (\ref{eq:model}).

Unfortunately, this problem is severely ill-posed without any additional constraints on $\bg_j$ since the number of samples in (\ref{eq:model}) is $N: = 4M+1$, while the number of unknowns in (\ref{eq:model}) is $O(JN)$, which is larger than $N$. To alleviate this, we solve our problem under the assumption that all $\bg_j$ live in a common low-dimensional subspace spanned by the columns of a known $N\times K$
matrix $\bB$ with $K\leq N$, which we denote as
\begin{equation*}
\bB = \begin{bmatrix}
\bb_{-2M} & \bb_{-2M+1} & \cdots & \bb_{2M-1} & \bb_{2M}
\end{bmatrix}^H
\end{equation*}
with $\bb_n \in \complex^{K\times 1}$. In other words, $\bg_j = \bB \bh_j$ for some unknown $\bh_j\in \complex^{K\times 1}$.  Henceforth, we assume that $\norm[2]{\bh_j} = 1$ without loss of generality. Under the subspace assumption, recovery of $\bg_j$ is guaranteed if $\bh_j$ can be recovered. Therefore, the number of degrees of freedom in (\ref{eq:model}) becomes $O(JK)$, which can possibly be smaller than the number of samples $N$ when $J, K \ll N$.

Under the subspace assumption, we can rewrite (\ref{eq:model}) as
\begin{equation}
\label{eq:submodel} \by(n) = \sum_{j =1}^J c_j e^{-i2\pi n \tau_j} \bb_n^H \bh_j.
\end{equation}
Defining
\begin{equation*}
\ba(\tau)  = \begin{bmatrix} e^{i2\pi (-2M)\tau} & \cdots & e^{i2\pi (0)\tau} & \cdots & e^{i2\pi (2M)\tau} \end{bmatrix}^T,
\end{equation*}
we have
\begin{equation}
\label{eq:MMVsensing}
\begin{aligned}
\by(n) & = \sum_{j =1}^J c_j \ba(\tau_j)^H \be_n \bb_n^H \bh_j \\
& = \trace{\left( \be_n \bb_n^H   \sum_{j =1}^Jc_j  \bh_j \ba(\tau_j)^H\right)}\\
& = \left\langle \sum_{j =1}^J c_j \bh_j \ba(\tau_j)^H , \bb_n\be_n^H \right\rangle,
\end{aligned}
\end{equation}
where we have defined $\langle \bX, \bY \rangle = \trace{(\bY^H \bX)}$ and used $\be_n$, $-2M \leq n \leq 2M$, to denote the $(n+2M+1)$th column of the $N \times N$ identity matrix $\bI_N$. We see that
(\ref{eq:MMVsensing}) leads to a parametrized rank-$J$ matrix sensing problem, which we write as
\begin{equation*}
\begin{aligned}
\by  & = \cB(\bX_o),
\end{aligned}
\end{equation*}
where the linear operator $\cB: \mathbb{C}^{K\times N} \rightarrow \mathbb{C}^{N}$ is defined as $[\cB(\bX_o)]_n = \left\langle \bX_o, \bb_n \be_n^H \right\rangle, n = -2M, \cdots, 2M$ with $\bX_o = \sum_{j =1}^J c_j \bh_j \ba(\tau_j)^H$. Here we choose the number of measurements $N = 4M+1$, which is purely for ease of theoretical analysis. We note that our result is not restricted to the symmetric case presented here and does not necessarily require that $N$ should be an odd number. We refer the interested reader to Appendix~A of \cite{tang2013compressed} for a discussion of how to modify the argument for the general case.

In many scenarios, the number of complex exponentials $J$ is small. Therefore, we use the atomic norm to promote sparsity. As in~\cite{chi2015guaranteed}, define the atomic norm \cite{chandrasekaran2012convex} associated with the following set of atoms
\begin{equation*}
\cA = \left\{ \bh \ba(\tau)^H:~\tau\in [0,1), \norm{\bh} = 1, \bh \in \complex^{K\times 1} \right\}
\end{equation*}
as
\begin{equation*}
\begin{aligned}
\norm[\cA]{\bX}  & = \inf \left\{ t>0:~\bX \in t {\bf conv}(\cA) \right\} \\
&  = \inf_{c_k, \tau_k, \norm{\bh_k} = 1}\left\{\sum_k |c_k|:~\bX = \sum_k c_k \bh_k \ba(\tau_k)^H \right\}.
\end{aligned}
\end{equation*}
To enforce the sparsity of the atomic representation, we solve
\begin{equation}
\label{eq:mainprog}
\begin{aligned}
& \mathop{\text{minimize}}\limits_{\bX}~~\norm[\cA]{\bX}   \\
&  {\text{subject~to}}~~\by(n) =\langle \bX,  \bb_n\be_n^H\rangle,~~n = -2M, \cdots, 2M.
\end{aligned}
\end{equation}
Standard Lagrangian analysis shows that the dual of (\ref{eq:mainprog}) is given by
\begin{equation}
\label{eq:dual}
\mathop{\text {maximize}}\limits_{\blambda}~~\langle \blambda, \by \rangle_{\mathbb{R}}~~~~~~{\text{ subject~to}}~~\| \mathcal{B}^{*}(\blambda)\|_{\mathcal{A}}^{*} \leq 1
\end{equation}
where $\langle \blambda, \by \rangle_{\real} = \text{Re} \left(\langle \blambda, \by \rangle\right)$, $\cB^{*}: \mathbb{C}^{N} \rightarrow \mathbb{C}^{K\times N}$ denotes the adjoint operator of $\cB$ and $\cB^{*}(\blambda) = \sum_{n} \blambda(n)\bb_n \be_n^H$, and $\| \cdot \|_{\mathcal{A}}^{*}$ is the dual norm of the atomic norm.

 The following proposition characterizes the optimality condition of program (\ref{eq:mainprog}) with the vector polynomial $\bq(\tau)$ serving as a dual certificate to certify the optimality of $\bX_o$ in the primal problem (\ref{eq:mainprog}).
\begin{prop}
Suppose that the atomic set $\mathcal{A}$ is composed of atoms of the form $\bh \ba(\tau)^H$ with $\|\bh\|_2 = 1, \tau \in [0,1)$. Define the set $\mathbb{D} = \left\{\tau_j, 1\le j\le J \right\}$. Let $\widehat{\bX}$ be the optimal solution to (\ref{eq:mainprog}). Then $\widehat{\bX} = \bX_o$ is the unique optimal solution if the following two conditions are satisfied: \\
1) There exists a dual polynomial
\begin{equation}
\label{eq:cond}
\begin{aligned}
\bq(\tau) & = \mathcal{B}^{*}(\blambda) \ba(\tau) \\
& = \sum_{n = -2M}^{2M} \blambda(n) e^{i2\pi n \tau} \bb_n
\end{aligned}
\end{equation}
satisfying   
\begin{equation}
\bq(\tau_j) = \sign(c_j) \bh_j, ~~~\forall~\tau_j \in \mathbb{D} \label{cond1}
\end{equation}
\begin{equation}
\|\bq(\tau)\|_2 <1,~~~\forall~\tau \notin \mathbb{D}. \label{cond2}
\end{equation}
Here $\blambda$ is a dual optimizer and $\sign(c_j) := \frac{c_j}{|c_j|}$.
\label{optimalitycond}\\
2) $\left\{ \begin{bmatrix} \vdots \\  \ba(\tau_j)^H \be_n \bb_n^H \\ \vdots \end{bmatrix}, j = 1, \cdots, J\right\}$ is a linearly independent set.
\end{prop}
We include the proof of Proposition \ref{optimalitycond} in Appendix \ref{sec:proofofmainthm}.

\subsection{SDP Characterization}

Since the convex hull of the set of atoms $\mathcal{A}$ can be characterized by a semidefinite program, $\| \bX \|_{\mathcal{A}}$ admits an equivalent SDP representation.
\begin{lemma} \cite{yang2014exact,li2014off} For any $\bX \in \mathbb{C}^{K\times N}$,
\begin{equation*}
\begin{aligned}
\| \bX\|_{\mathcal{A}} & = \inf_{\bu,
\bT}   \left\{ \frac{1}{2N} \trace \left(\Toep(\bu)\right) + \frac{1}{2}\trace(\bT): 
\begin{bmatrix}
\Toep(\bu) & \bX^H \\
\bX & \bT
\end{bmatrix}   \succeq 0
\right\},
\end{aligned}
\end{equation*}
where $\bu$ is a complex vector whose first entry is real, $\Toep(\bu)$ denotes the $N \times N$ Hermitian Toeplitz matrix whose first column is $\bu$, and $\bT$ is a Hermitian $K\times K$ matrix.
\end{lemma}
Hence, (\ref{eq:mainprog}) can be solved efficiently using off-the-shelf SDP solvers such as CVX \cite{grant2008cvx}. With many SDP solvers, one can also obtain a dual optimal solution $\blambda$ to (\ref{eq:dual}) for free by solving the primal program (\ref{eq:mainprog}), and as we discuss below, this can be used to localize the supports of the point sources. We note that the dual optimal solution is not unique in general. As discussed in  \cite{tang2013compressed}, the recovered support set from the dual solution must contain the true support set when the optimal primal solution is $\bX_o$. Though it is possible that the recovered support set contains spurious parameters, solving the SDP with the interior point method will avoid this pathological situation and recover the true support exactly. See \cite{tang2013compressed} for more technical discussions on this.

Now, given the dual optimal solution $\blambda$, consider the trigonometric polynomial:
\begin{equation*}
\begin{aligned}
  p(e^{i2\pi \tau})
 &  = 1 - \|\bq(\tau)\|_2^2\\
 &  = 1 - \bq(\tau)^H \bq(\tau) \\
& = 1 - \sum_{n = -4M}^{4M} u_n e^{i2\pi n \tau},
\end{aligned}
\end{equation*}
where $\bq(\tau)$ is defined in terms of $\blambda$ in~\eqref{eq:cond}, and $u_n$ are some scalars that can be computed from $\bq(\tau)$ explicitly. To localize the supports of the point sources, one can simply compute the roots of the polynomial $p(z)$ on the unit circle. This method allows for the recovery of point sources to very high precision as shown in \cite{candes2014towards}.

Another way to recover the support is by discretizing $\tau \in [0,1)$ on a fine grid up to a desired accuracy. Then, one can check the $\ell_2$ norm of the dual polynomial $\bq(\tau)$ and identify the $\{\tau_j\}$ by selecting the values of $\tau$ such that $\norm[2]{\bq(\tau)} = \|\sum_{n = -2M}^{2M} \blambda(n) e^{i2\pi n \tau} \bb_n \| \approx 1$. We use this heuristic in our numerical simulations.

Given an estimate of $\{\tau_j\}$, say $\{\widehat{\tau}_j\}$, one can plug these values back into~\eqref{eq:MMVsensing} to form the following overdetermined linear system of equations:
\begin{equation*}
\begin{aligned}
& \begin{bmatrix}
 \ba(\widehat{\tau}_1)^H  \be_{-2M} \bb_{-2M}^H & \cdots &  \ba(\widehat{\tau}_J)^H  \be_{-2M} \bb_{-2M}^H \\
\vdots  & \ddots &  \vdots \\
\ba(\widehat{\tau}_1)^H  \be_{2M} \bb_{2M}^H & \cdots &  \ba(\widehat{\tau}_J)^H  \be_{2M} \bb_{2M}^H \\
\end{bmatrix} \begin{bmatrix} c_1 \bh_1 \\ \vdots \\ c_J \bh_J \end{bmatrix}  =   \begin{bmatrix} \by(-2M) \\ \vdots \\ \by(2M)  \end{bmatrix}.
\end{aligned}
\end{equation*}
A unique solution for $\{c_j\bh_j\}_{j = 1}^J$ can be obtained by a simple least squares since the columns of the above matrix are linearly independent. However, we note that one cannot resolve the inherent scaling ambiguity between each $c_j$ and the corresponding $\bh_j$.

\section{Recovery Guarantee}
\label{sec:recovery guarantees}
\subsection{Sample Complexity Bound for Exact Recovery}
Given samples of the form $\by(n) = \langle \bX_o, \bb_n \be_n^H \rangle, n = -2M,\dots, 2M$, we wish to quantify precisely how large the number of measurements $N = 4M+1$ should be so that the atomic norm minimization (\ref{eq:mainprog})
exactly recovers $\bX_o$.

Before presenting the main result of our work, we discuss the assumptions that are used in the main theorem. The assumptions can be grouped into three categories: (a)  randomness and incoherence of the subspace spanned by the columns of $\bB$, (b)  minimum separation of the $\tau_j$, and (c) uniform distribution of $\bh_j$ on the complex unit sphere $\mathbb{CS}^{K-1}$.

We assume that the columns of the matrix $\bB^H$, namely, $\bb_n, -2M\leq n \leq 2M$,  are independently sampled from a population $\cF$ with the following properties \cite{candes2011probabilistic}:
\begin{itemize}
\item {\bf Isotropy  property:} We assume that the distribution $\cF$ obeys the isotropy property in that
\begin{equation}
\label{eq:isotropyproperty}
\mathbb{E} \bb \bb^H = \bI_K,~~~ \bb \sim \cF.
\end{equation}
\item {\bf Incoherence property:} We assume that $\cF$ satisfies the incoherence property with coherence $\mu$ in that
\begin{equation}
\label{eq:incoherenceproperty}
\max_{1\leq p\leq K} |\bb(p)|^2 \leq \mu,~~~ \bb \in \mathcal{F}
\end{equation}
holds almost surely, where $\bb(p)$ is the $p$th element of $\bb$. For ease of theoretical analysis, hereafter we also assume that $\mu K\geq 1$, which can always be ensured by choosing $\mu$ sufficiently large. In particular, when the incoherence property (\ref{eq:incoherenceproperty}) holds almost surely, the isotropy property (\ref{eq:isotropyproperty}) ensures that $\mu \geq 1$ and thus that $\mu K\geq 1$.
\end{itemize}

Furthermore, we require the following conditions on the parameters of the complex exponentials and the rotations of $\bg_j$ in the subspace $\bB$, namely, $\bh_j$.
\begin{itemize}
\item {\bf Minimum separation:} We assume that
\begin{equation*}
\Delta_{\tau} = \min_{k \neq j} |\tau_k - \tau_j| \geq \frac{1}{M}
\end{equation*}
where the distance $|\tau_k - \tau_j|$ is understood as the wrap-around distance on $[0,1)$.
\item {\bf Random rotation:} We assume that the coefficient vectors $\bh_j$ are drawn i.i.d.\ from the uniform distribution on the complex unit sphere $\mathbb{CS}^{K-1}$.
\end{itemize}
\begin{thm}
\label{mainthm}
Assume that the minimum separation condition $\Delta_{\tau} \geq \frac{1}{M}$ is satisfied and that $M \geq 64$. Also, assume that $\bg_{j} = \bB \bh_j$ with the columns of $\bB^H$, namely, $\bb_n$, being i.i.d.\ samples from a distribution $\mathcal{F}$ that satisfies the isotropy  and incoherence properties with coherence parameter $\mu$. Additionally, assume that $\bh_j$ are drawn i.i.d.\ from the uniform distribution on the complex unit sphere $\mathbb{CS}^{K-1}$. Then, there exists a numerical constant $C$ such that
\begin{equation}
\label{eq:samplebound}M \geq C \mu J K\log\left(\frac{MJK}{ \delta}\right) \log^2\left(\frac{MK}{ \delta}\right)
\end{equation}
is sufficient to guarantee that we can recover $\bX_o$ via program (\ref{eq:mainprog}) with probability at least $1-\delta$.
\end{thm}

\subsection{Discussion}
\label{sec:discussmainresult}

Inspired by~\cite{chi2015guaranteed}, we use the same assumptions on the random subspace model (the isotropy and incoherence properties) in order to prove our Theorem~\ref{mainthm}. The randomness assumption on the subspace does not appear to be critical in practice, as evidenced by our numerical experiments in Section~\ref{sec:Numericalsimulations}; being able to replace this with a deterministic condition would increase the relevance of our theory to the applications discussed in the introduction.

Also, as noted in the introduction, our work generalizes the model in \cite{chi2015guaranteed} to the non-stationary case. It may also be possible to extend the result developed in \cite{chi2015guaranteed} to the non-stationary case; however, the sample complexity would still be $O(J^2K^2)$, up to a polylogarithmic factor. In contrast, we reduce the sample complexity to $O(JK)$, which is information theoretically optimal, up to a polylogarithmic factor. In order to do this, in the proof of Lemma \ref{lem:matroperatornormbound} in Section \ref{sec:detailedproofofmainthm}, we apply a matrix Bernstein's inequality instead of Talagrand's concentration inequality which was used by \cite{chi2015guaranteed}. Our theorem also relies on an additional assumption that was not present in~\cite{chi2015guaranteed}, namely that the coefficient vectors $\bh_j$ are drawn randomly. We do not believe that this randomness assumption is important in practice and suspect that it is merely an artifact of our proof.

Our bound on $M$ suggests that when $\mu$ is a constant (e.g., when the rows of $\bB$ are drawn from a sub-Gaussian distribution, $\mu$ can be bounded by a constant times $\log K$ with high probability \cite{candes2011probabilistic}), $M = O(JK)$ is sufficient for exact recovery and this matches the number of degrees of freedom in the problem, up to a polylogarithmic factor. Thus, our sample complexity bound is tight and there is little room for further improvement. When the dimension of the subspace is bounded by a constant, $M = O(J)$ (up to a polylogarithmic factor) is sufficient for exact recovery. This bound matches the one in the deterministic super-resolution framework \cite{candes2014towards}, where $N = O(J)$ suffices to exactly localize the unknown spikes under the same minimum separation condition used here.
We also see that our bound improves the one derived in \cite{chi2015guaranteed} even when $\bg_{j} = \bB \bh$, i.e., when $\bg_{j}$ has no dependence on $j$. We note that the number of degrees of freedom in that problem is $O(J+K)$. It would be interesting to see if further improvement upon our bound is possible in the stationary scenario.

Finally, when the measurements are contaminated by noise, one can extend the observation model as
\begin{equation*}
\by =   \cB(\bX_o) + \bz,
\end{equation*}
where $\|\bz\|_2 \leq \delta_{\text{noise}}$ and $\delta_{\text{noise}}$ is a parameter controlling the noise level. To recover an estimate of $\bX_o$, we propose to solve the following inequality constrained program:
\begin{equation}
\label{eq:noisecase}
 \mathop{\text{minimize}}\limits_{\bX}~\norm[\cA]{\bX} ~~{\text{ subject~to}}~~\| \by -\cB(\bX) \|_2 \leq \delta_{\text{noise}}.
\end{equation}
We leave robust performance analysis of (\ref{eq:noisecase}) for future research.

\section{Numerical Simulations}
\label{sec:Numericalsimulations}

\begin{figure*}[htb]
  \centering
  \begin{subfigure}{.49\textwidth}
  \centering
  \centerline{\includegraphics[width=1\linewidth]{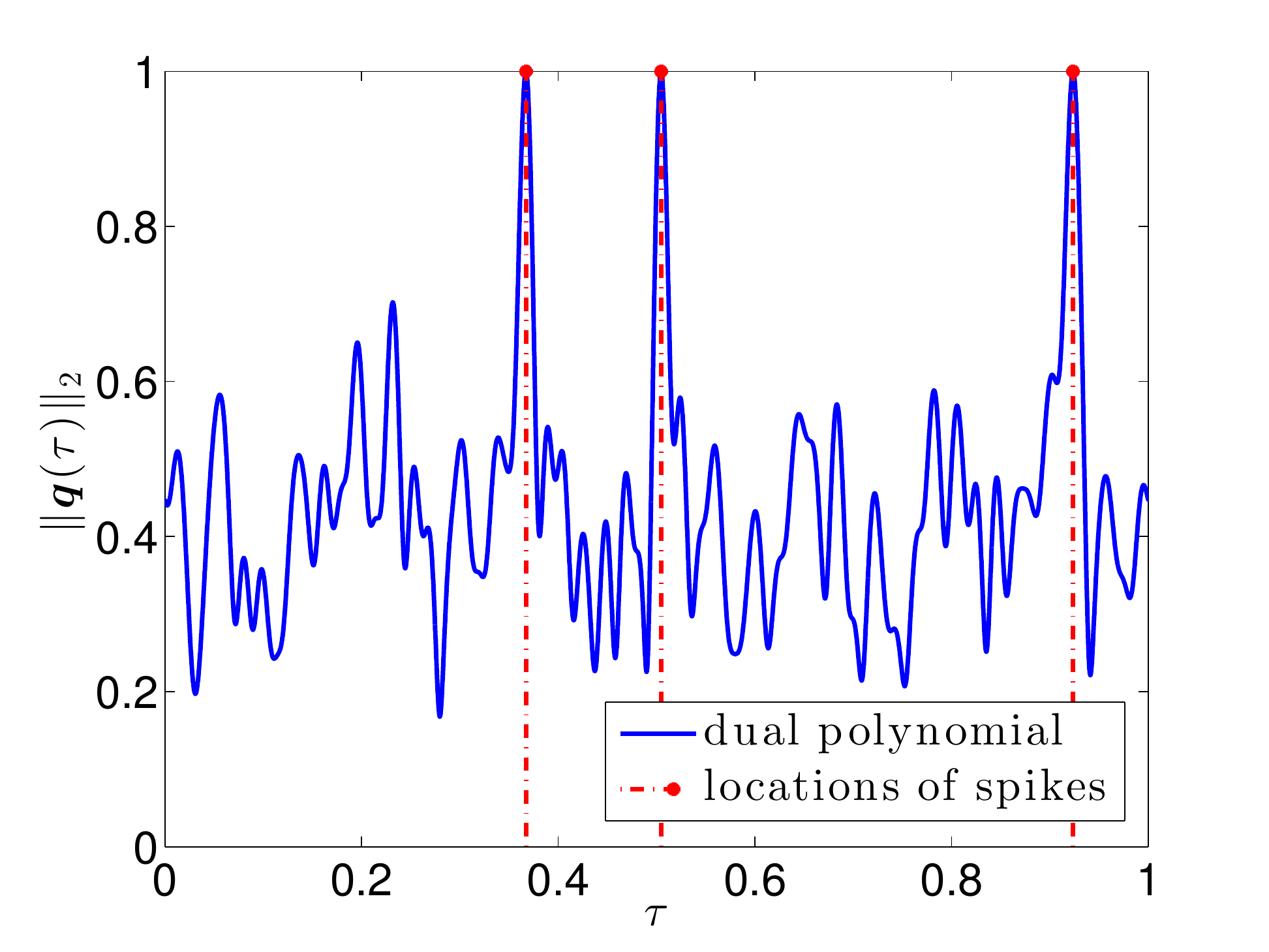}}
  \caption{}
  \label{fig:sub1}
\end{subfigure}
\begin{subfigure}{.49\textwidth}
\centering
  \centerline{\includegraphics[width=1\linewidth]{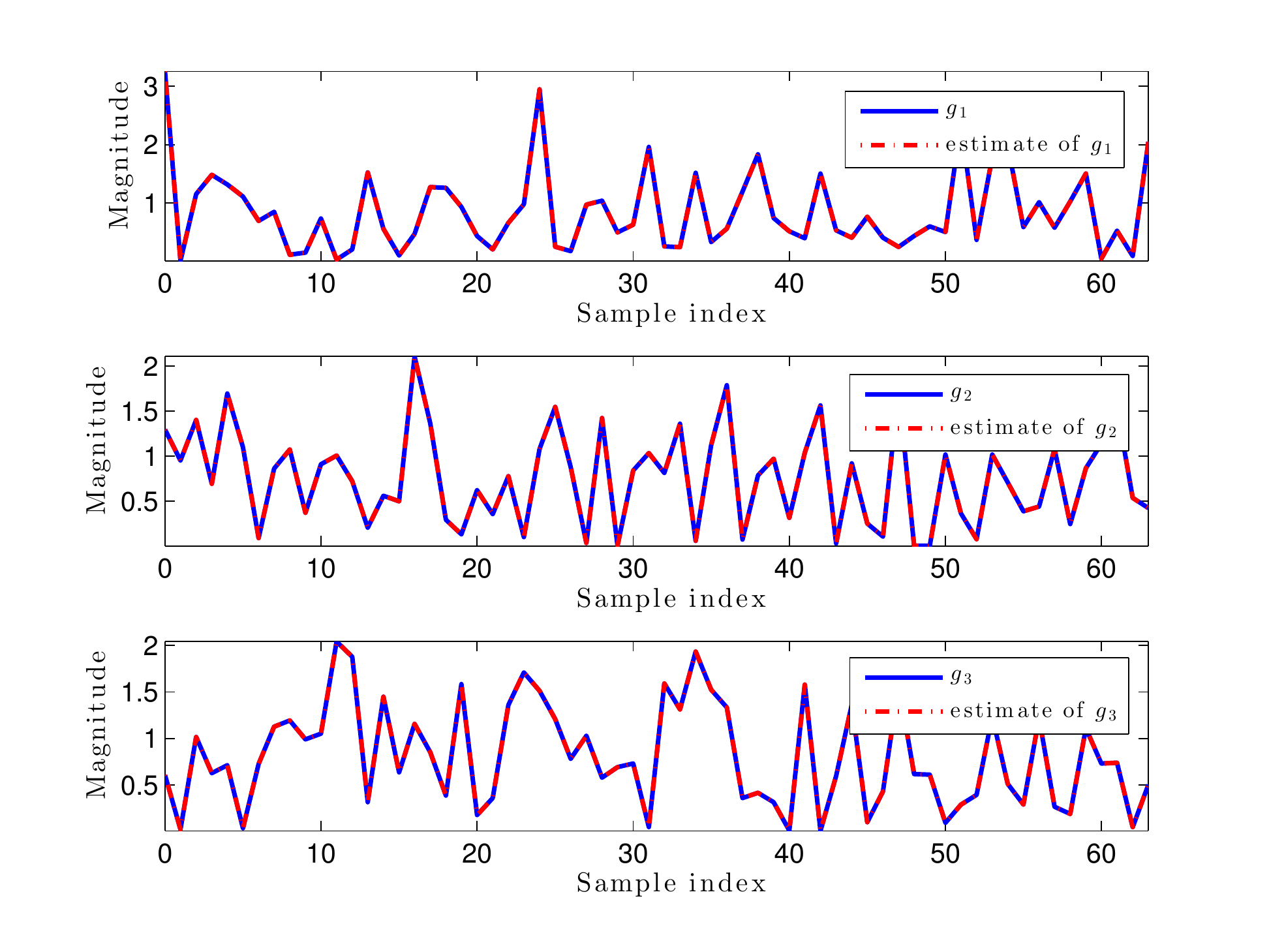}}
   \caption{}
  \label{fig:sub2}
 \end{subfigure}
\caption{(a) The $\ell_2$-norm of the dual polynomial $\norm[2]{\bq(\tau)}$ and the locations of the true spikes when entries of $\bB$ are built from the standard real Gaussian distribution. (b) The magnitude of samples of the waveforms $\bg_1, \bg_2, \bg_3$ and their estimates $\widehat{\bg}_1, \widehat{\bg}_2, \widehat{\bg}_3$ from least squares (best viewed in color).}
\label{fig:figure1}
\end{figure*}

\begin{figure*}[htb]
  \centering
  \begin{subfigure}{.32\textwidth}
  \centering
  \centerline{\includegraphics[width=1\linewidth]{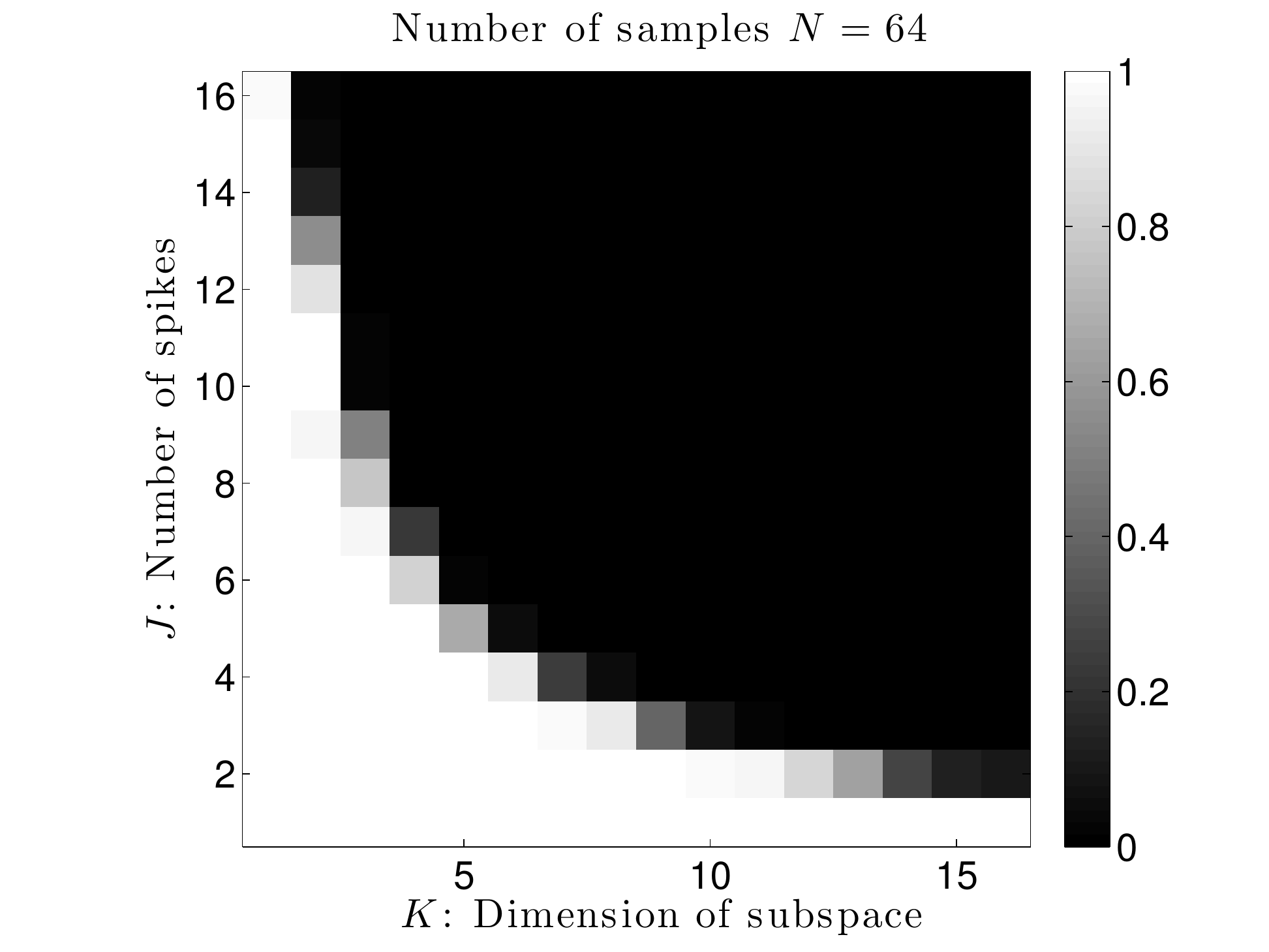}}
  \caption{}
  \label{fig3:sub1}
\end{subfigure}
\begin{subfigure}{.32\textwidth}
\centering
  \centerline{\includegraphics[width=1\linewidth]{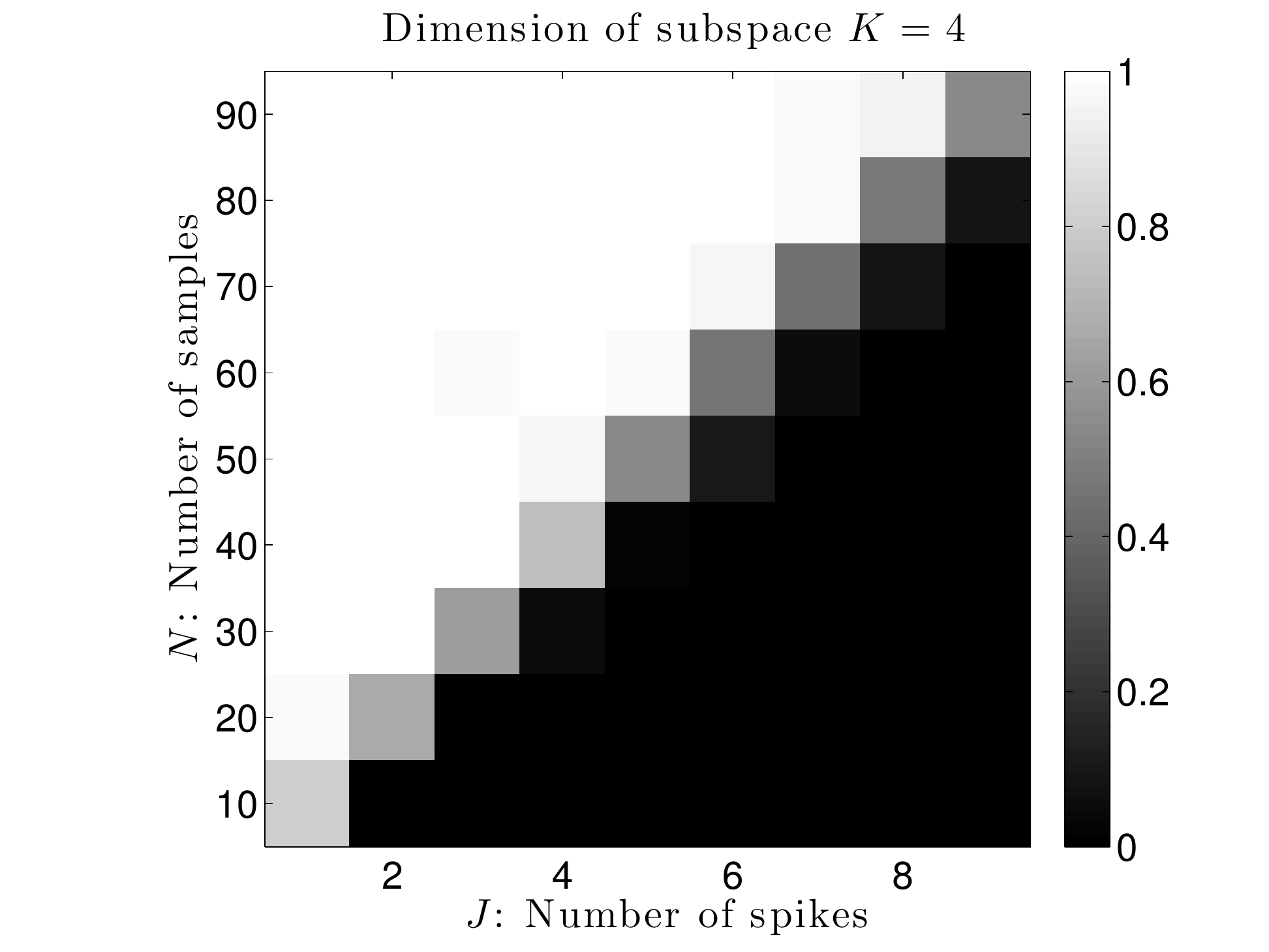}}
    \caption{}
  \label{fig3:sub2}
 \end{subfigure}
 \begin{subfigure}{.32\textwidth}
\centering
  \centerline{\includegraphics[width=1\linewidth]{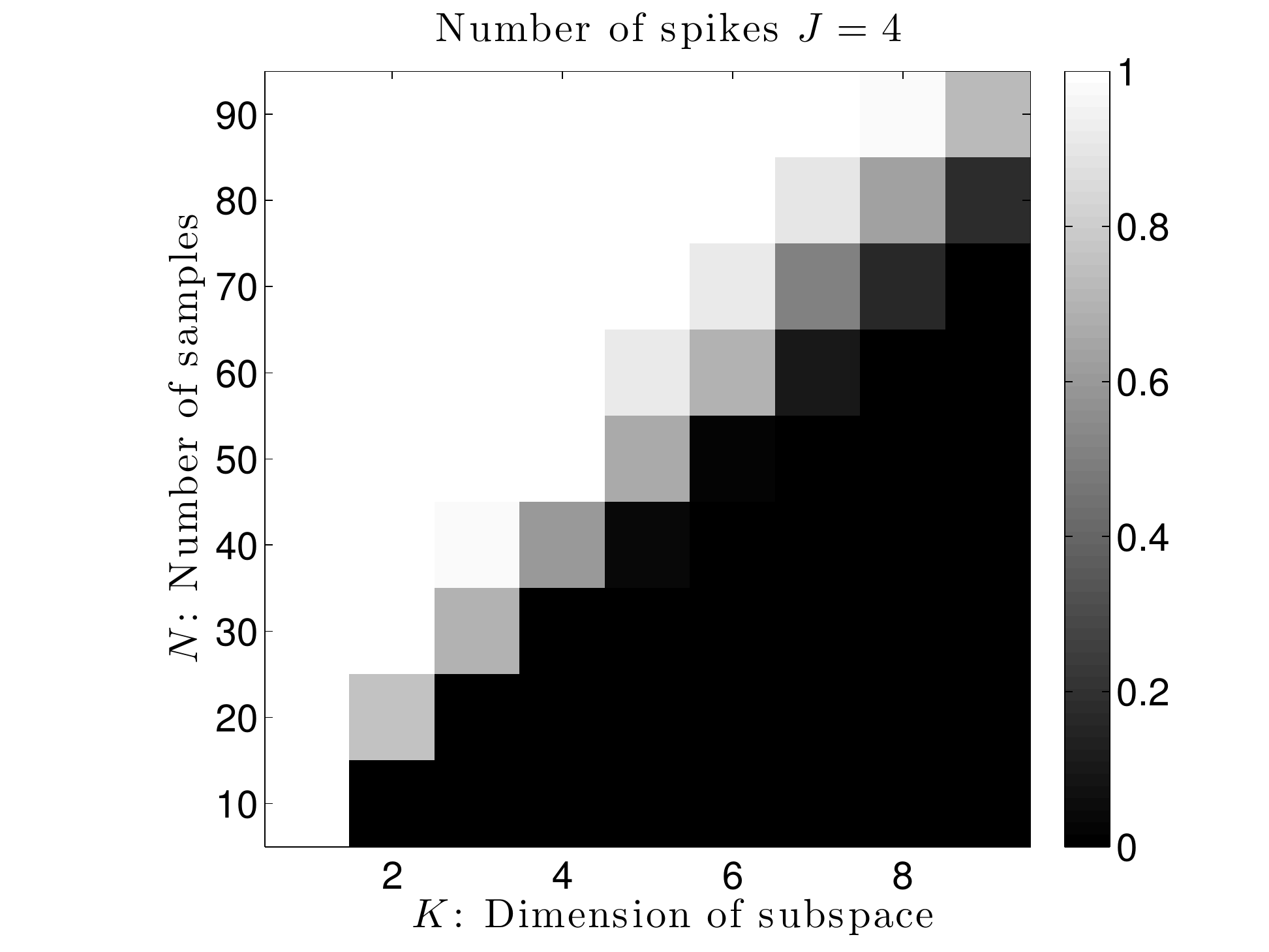}}
    \caption{}
  \label{fig3:sub3}
 \end{subfigure}
 \caption{(a) The probability of success of non-stationary blind super-resolution using atomic norm minimization when $N = 64$ is fixed. (b) The phase transition of non-stationary blind super-resolution using atomic norm minimization when the dimension of subspace is fixed, $K = 4$. (c) The phase transition of non-stationary blind super-resolution using atomic norm minimization when the number of spikes is fixed, $J = 4$.}
 \label{fig:figure2}
\end{figure*}

\begin{figure*}
  \centering
  \begin{subfigure}{.49\textwidth}
  \centering
  \centerline{\includegraphics[width=1\linewidth]{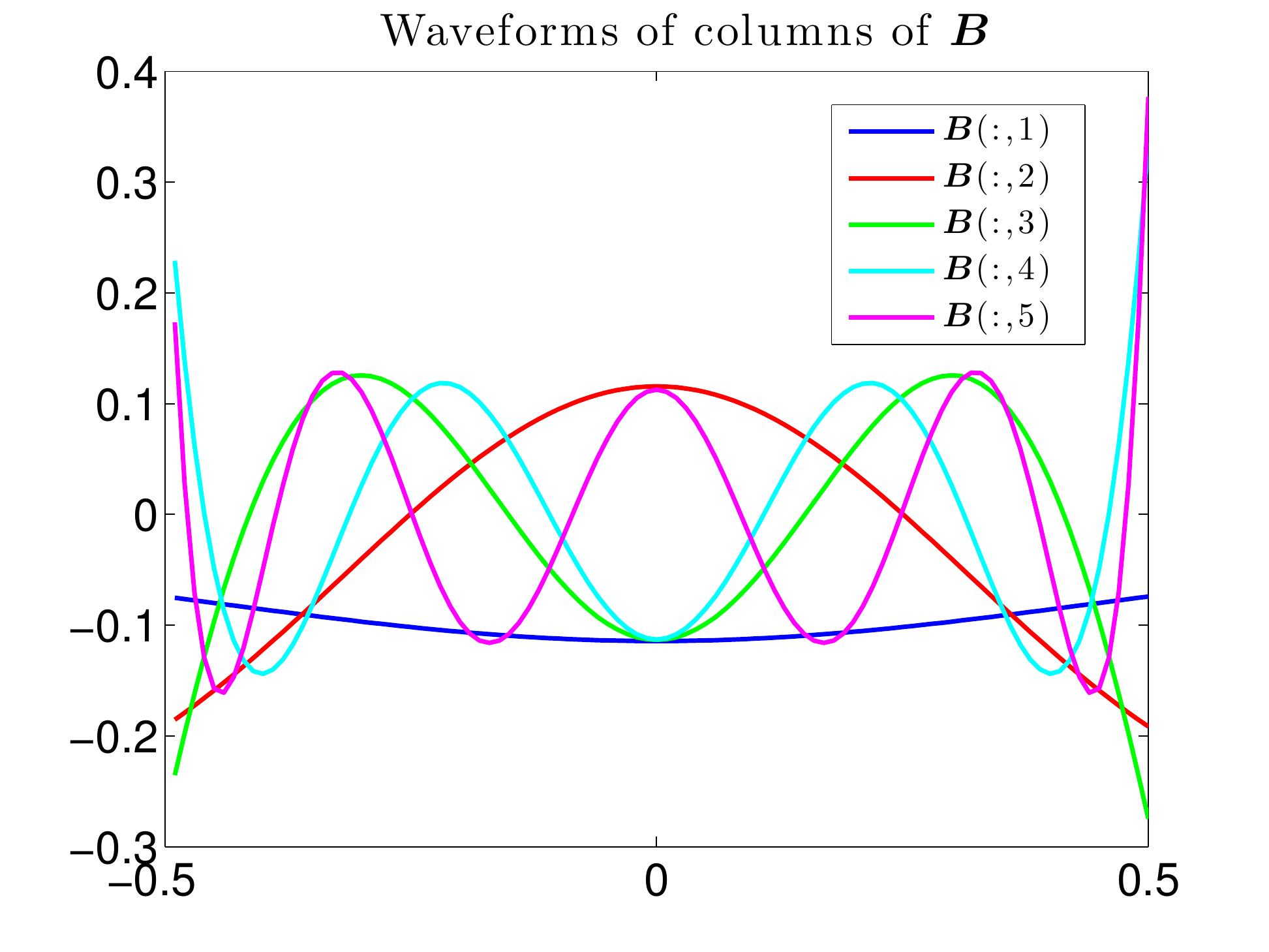}}
  \caption{}
  \label{fig4:sub1}
\end{subfigure} 
\begin{subfigure}{.49\textwidth}
\centering
  \centerline{\includegraphics[width=1\linewidth]{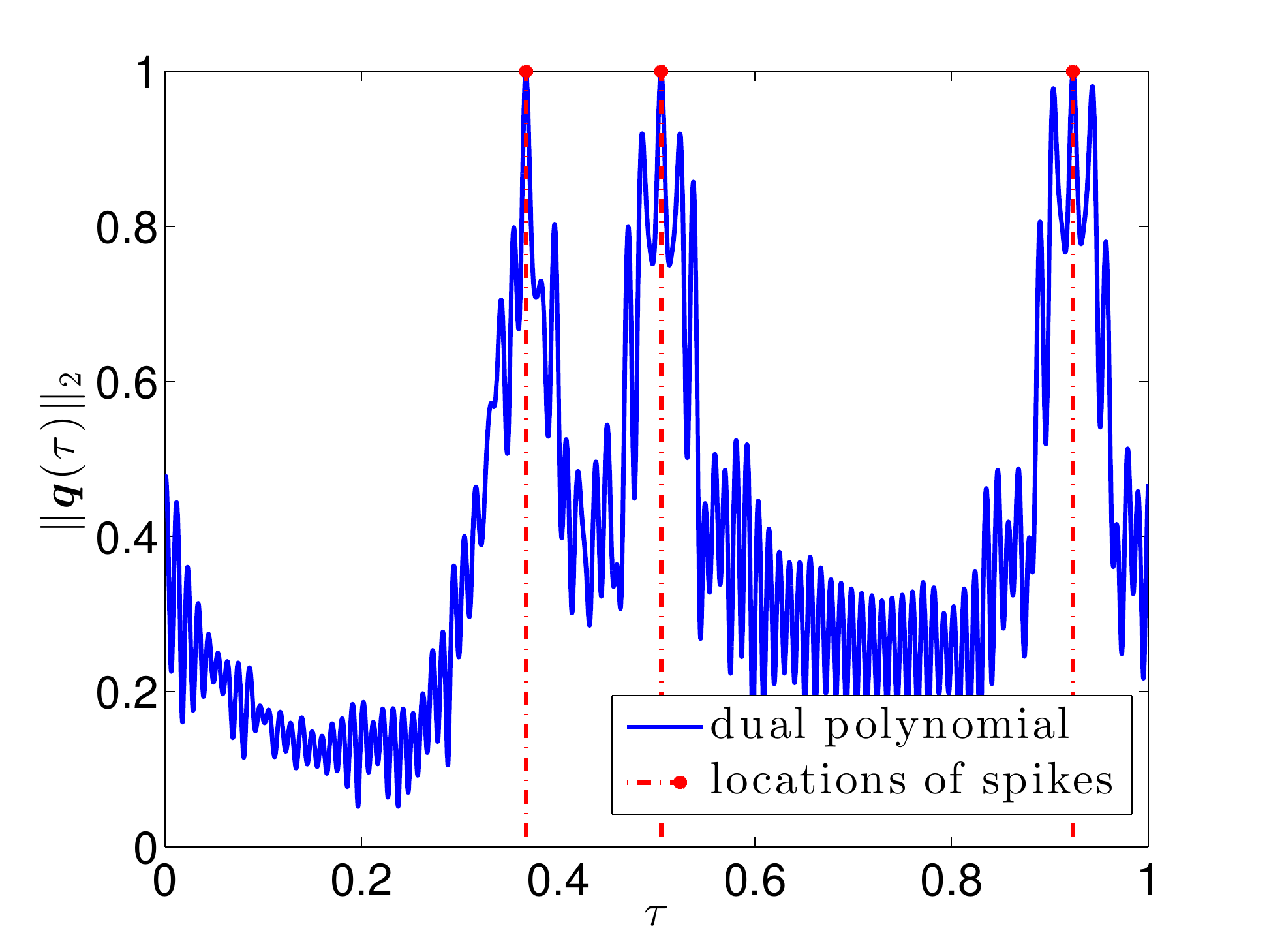}}
    \caption{}
  \label{fig4:sub2}
 \end{subfigure}
 \caption{(a) Plot of orthonormal columns of $\bB$. (b)  The $\ell_2$-norm of the dual polynomial $\norm[2]{\bq(\tau)}$ and the locations of the true spikes when $\bB$ is built from the left singular vectors of the low-rank approximation of the dictionary $\bD_{\bg}$.}
 \label{fig:figure3}
\end{figure*}

In this section, we provide synthetic numerical simulations to
support our theoretical findings. In all of the simulations, we solve the atomic norm minimization problem (\ref{eq:mainprog}) using CVX. We start with a simple example to illustrate how to localize unknown spikes using the dual polynomial $\bq(\tau)$ and recover samples of the unknown waveforms $\bg_j(n)$. For the sake of illustration, we set $N = 64$, $J = 3$ and $K = 4$ and randomly generate the locations of $J$ spikes on $[0,1)$ from a uniform distribution. We regenerate the set of locations if a particular minimum separation condition is violated; in particular, we ensure that the minimum separation $\Delta_{\tau}$ between spikes is not less than $\frac{1}{N}$, which is slightly smaller than $\frac{1}{M}$ required by our theorem. Each $c_j$ is generated randomly with a dynamic range of $20$dB and uniform phase. We build $\bB$ with entries generated randomly from the standard Gaussian distribution. Each $\bh_j$ is also generated using i.i.d.\ real Gaussian random variables and is then normalized. Figure \ref{fig:figure1}(a) shows the $\ell_2$-norm of the dual polynomial, namely, $\norm[2]{\bq(\tau)}$ on a discretized set of points on [0,1] with discretization step size $10^{-4}$.
Once the $\tau_j$'s are identified, we solve a least squares problem using a pseudo-inverse to estimate the $\bh_j$'s. Because we cannot exactly recover the $\bh_j$'s, hence the $\bg_j$'s,  due to phase rotations, we plot the magnitude of the $\bg_j$'s and the estimated $\widehat{\bg}_j$'s in Figure \ref{fig:figure1}(b). We also compute $|\langle \bh_j, \widehat{\bh}_j \rangle|, j = 1, 2, 3$, which equal $0.99999984, 0.99999973$ and $0.99999996$, respectively. For $j = 1, 2, 3$, the agreements on the magnitude of $\bg_j$ and $\widehat{\bg}_j$ and on the absolute inner product between $\bh_j$ and $\widehat{\bh}_j$ confirm that we can faithfully recover $\bg_j$.

Next, we characterize the phase transition of atomic norm minimization (\ref{eq:mainprog}). We run $50$ trials for each pair of the underlying changing variables. For each trial, we declare success if the relative reconstruction Frobenius norm error of $X_o$ is less than $10^{-4}$. In the first experiment, we fix $N = 64$ and vary the values of $J$ and $K$. Other specifications are the same as the warm-up experiment above.  Figure \ref{fig:figure2}(a) shows the phase transition in this situation. From Figure \ref{fig:figure2}(a), we can roughly see that the success transition curve behaves like a hyperbola, which matches the bound appearing in Theorem \ref{mainthm}. In the second experiment, we fix the dimension of the subspace $K = 4$ and study the phase transition between $N$ and $J$. Figure \ref{fig:figure2}(b) indicates a linear relationship between $N$ and $J$ when $K$ is fixed. Finally, we test the phase transition between $N$ and $K$ when the number of spikes $J$ is fixed. We set $J = 4$ and vary $N$ and $K$ in the experiment. Figure \ref{fig:figure2}(c) shows the phase transition in this situation and also implies a linear relationship between $N$ and $K$.

Finally, we test the atomic norm minimization (\ref{eq:mainprog}) for localization of spikes in a more practical scenario where the matrix $\bB$ is generated by extracting the principal components of a structured matrix, and thus has less randomness. To set up the experiment, we set $J = 3$ and generate the locations of $\{\tau_j\}$ uniformly at random between $0$ and $1$ under the minimum separation $\Delta_{\tau} = \frac{4}{N}$, which is roughly the same as what Theorem \ref{mainthm} requires. We assume that $\bg_j(n)$ are samples of the zero mean Gaussian waveform $g_{\sigma^2}(t) = \frac{1}{\sqrt{2\pi \sigma^2}} e^{-\frac{t^2}{2\sigma^2}}$ with unknown variance $\sigma^2 \in [0.1,1]$. We take $N = 100$ samples around the origin with sampling interval $\Delta t = 0.01$. To obtain the matrix $\bB$, we first build a dictionary as follows:
\begin{equation*}
\bD_{\bg} = \begin{bmatrix} \bg_{\sigma^2 = 0.1} & \bg_{\sigma^2= 0.11} & \bg_{\sigma^2 = 0.12} & \cdots &  \bg_{\sigma^2 = 1}  \end{bmatrix},
\end{equation*}
where columns of $\bD_{\bg}$ are from the samples of Gaussians with a discretized set of variances $\sigma^2 \in [0.1,1]$. Then, we apply PCA on
$\bD_{\bg}$ and obtain the best rank-$5$ approximation, which we denote as $\bD_{\bg,5}$. Finally, we construct $\bB$ by taking its columns to be the left singular vectors of $\bD_{\bg,5}$. Figure \ref{fig:figure3}(a) depicts the waveforms of the columns of $\bB$. We note that $\bD_{\bg,5}$ gives a very good approximation to $\bD_{\bg}$ due to the sharp decay of the singular values of $\bD_{\bg}$. In particular, $\norm[F]{\bD_{\bg}-\bD_{\bg,5}}/\norm[F]{\bD_{\bg}} = 1.9860\times 10^{-8}$. We generate each $\bg_j(n)$  by choosing its variance uniformly at random between $0.1$ and $1$, and then we build the measurements $\{\by(n)\}$. In particular, we note that $\{\bg_j\}$ do not necessarily belong to the columns of $\bD_{\bg}$. Finally, we solve (\ref{eq:mainprog}) using CVX. Figure \ref{fig:figure3}(b) shows the $\ell_2$-norm of the dual polynomial, namely, $\norm[2]{\bq(\tau)}$, on a discretized set of points on [0,1] with discretization step size $10^{-4}$. We can see that one can still localize $\{\tau_j\}$ even when there is model mismatch between the subspace spanned by the columns of $\bB$ and $\bg_j(n)$.

\section{Proof of Theorem~\ref{mainthm}}
\label{sec:detailedproofofmainthm}
In this section, we prove Theorem \ref{mainthm}, the main result of this paper. We divide the proof into three subsections. In Subsection \ref{subsec:subsection1}, we construct the pre-certificate dual polynomial $\bq(\tau)$. In Subsections \ref{subsec: subsection2} and \ref{subsec:subsection3}, we show that our constructed dual polynomial $\bq(\tau)$ satisfies conditions (\ref{cond1}) and (\ref{cond2}), respectively.
\subsection{Construction of the Dual Polynomial $\bq(\tau)$}
\label{subsec:subsection1}
According to Proposition \ref{optimalitycond}, our goal is to explicitly construct a dual polynomial $\bq(\tau)$ such that conditions (\ref{cond1}) and (\ref{cond2}) hold. To proceed, we require that for all $\tau_j \in \mathbb{D}$
\begin{equation}
\label{subcond1}
\bq(\tau_j) = \sign (c_j) \bh_j,
\end{equation}
\begin{equation}
\label{subcond2}
-\bq'(\tau_j) = \bzero_{K\times 1}.
\end{equation}
Note that (\ref{subcond1}) is exactly the same as (\ref{cond1}), while (\ref{subcond1}) and (\ref{subcond2}) together form a necessary condition for (\ref{cond2}) to hold.
We construct the dual polynomial in (\ref{eq:cond}) by finding a proper $\blambda$ from solving the following weighted least energy minimization program with diagonal weighting matrix $\bW = \diag\left(\begin{bmatrix}w_{-2M} & \cdots &  w_{2M} \end{bmatrix}\right)$ to be determined later, $w_n>0, -2M \leq n \leq 2M$:
\begin{equation}
\label{eq:weightedleastnorm}
\begin{aligned}
& \mathop{\text{ minimize}}\limits_{\blambda}~~~~\norm[2]{\bW \blambda}^2   \\
&  {\text{ subject~to}}~~~\bq(\tau_j) = \sign (c_j) \bh_j,~~j = 1, \cdots, J,\\
&~~~~~~~~~~~~~~~~~ -\bq'(\tau_j) = \bzero_{K\times 1},~~j = 1, \cdots, J.
\end{aligned}
\end{equation}
We can see that the equality constraints in (\ref{eq:weightedleastnorm}) can be written as a linear system of equations (\ref{eq:largeeqation}). 
\begin{figure*}[!t]
\normalsize
\begin{equation}
\begin{aligned}
\label{eq:largeeqation}
\begin{bmatrix}
\bb_{-2M} \be_{-2M}^H \ba(\tau_1) & \cdots  & \bb_{2M} \be_{2M}^H \ba(\tau_1) \\
\vdots & \ddots & \vdots \\
\bb_{-2M} \be_{-2M}^H \ba(\tau_J) & \cdots  & \bb_{2M} \be_{2M}^H \ba(\tau_J) \\
-i2\pi (-2M)\bb_{-2M} \be_{-2M}^H \ba(\tau_1) & \cdots  & -i2\pi (2M)\bb_{2M} \be_{2M}^H \ba(\tau_1) \\
\vdots & \ddots & \vdots \\
-i2\pi (-2M)\bb_{-2M} \be_{-2M}^H \ba(\tau_J) & \cdots  & -i2\pi (2M)\bb_{2M} \be_{2M}^H \ba(\tau_J) \\
\end{bmatrix}
 \begin{bmatrix} \blambda(-2M) \\
\vdots \\
 \blambda(2M)
\end{bmatrix} = \begin{bmatrix} \sign(c_1)\bh_1 \\
\vdots \\
\sign(c_J) \bh_J \\
\bzero_{K\times 1} \\
\vdots \\
\bzero_{K\times 1} \end{bmatrix}.
\end{aligned}
\end{equation}
\hrulefill
\end{figure*}
For notational simplicity, we denote equation (\ref{eq:largeeqation}) as
\begin{equation*}
\begin{aligned}
\bA \blambda = \bp.
\end{aligned}
\end{equation*}
The optimality condition ensures that the optimal $\blambda$ has the form
\begin{equation*}
\begin{aligned}
\blambda & = \left(\bW^H \bW\right)^{-1} \bA^H \begin{bmatrix} \balpha \\
\bbeta \end{bmatrix} \\
& = \left(\bW^H \bW\right)^{-1} \left( \sum_{j = 1}^J \begin{bmatrix} \ba(\tau_j)^H \be_{-2M}\bb_{-2M}^H \\
\vdots \\
\ba(\tau_j)^H \be_{2M}\bb_{2M}^H\end{bmatrix} \balpha_j +  \sum_{j = 1}^J  \begin{bmatrix} i2\pi (-2M) \ba(\tau_j)^H \be_{-2M}\bb_{-2M}^H \\
\vdots \\
i2\pi (2M)\ba(\tau_j)^H \be_{2M}\bb_{2M}^H\end{bmatrix} \bbeta_j\right),
\end{aligned}
\end{equation*}
for some
\begin{equation*}
\begin{aligned}
\begin{bmatrix} \balpha \\
\bbeta \end{bmatrix} & = \begin{bmatrix}
\balpha_1^H &
\cdots &
\balpha_J^H &
\bbeta_1^H &
\cdots &
\bbeta_J^H
\end{bmatrix}^H, \balpha_j, \bbeta_j \in \mathbb{C}^{K\times 1}
\end{aligned}
\end{equation*}
satisfying the normal equation
\begin{equation*}
 \bA \left(\bW^H \bW\right)^{-1} \bA^H \begin{bmatrix} \balpha \\
\bbeta \end{bmatrix} = \bp.
\end{equation*}
Inserting $\blambda$ back into (\ref{eq:cond}), we obtain the dual polynomial
\begin{equation*}
\begin{aligned}
\bq(\tau) & = \sum_{n = -2M}^{2M} \blambda(n) e^{i2\pi n \tau} \bb_n  \\
& = \sum_{j = 1}^J \left(\sum_{n = -2M}^{2M} \frac{1}{w_n^2} e^{-i2\pi n\tau_j } \bb_n^H \balpha_j e^{i2\pi n\tau} \bb_n +  \sum_{n = -2M}^{2M} \frac{i2\pi n}{w_n^2} e^{-i2\pi n\tau_j } \bb_n^H \bbeta_j e^{i2\pi n\tau} \bb_n\right) \\
& = \sum_{j = 1}^J  \left(\sum_{n = -2M}^{2M} \frac{1}{w_n^2} e^{i2\pi n (\tau-\tau_j)} \bb_n \bb_n^H \balpha_j  +  \sum_{n = -2M}^{2M} \frac{i2\pi n}{w_n^2} e^{i2\pi n(\tau-\tau_j)} \bb_n \bb_n^H \bbeta_j  \right)\\
& = \sum_{j = 1}^J \bK_M (\tau - \tau_j) \balpha_j + \sum_{j =1}^J \bK_{M}'(\tau-\tau_j) \bbeta_j,
\end{aligned}
\end{equation*}
where
\begin{equation*}
\bK_M(\tau) =\sum_{n = -2M}^{2M} \frac{1}{w_n^2} e^{i2\pi n\tau} \bb_n \bb_n^H
\end{equation*}
and
\begin{equation*}
\bK'_M(\tau) = \sum_{n = -2M}^{2M} \frac{i2\pi n}{w_n^2} e^{i2\pi n\tau} \bb_n \bb_n^H,
\end{equation*}
which is the entry-wise derivative of $\bK_M(\tau)$ with respect to $\tau$. In the following, we choose $w_n = \sqrt{\frac{M}{g_M(n)}}$, where $g_M(n)$ is the convolution of two triangle functions given below:
\begin{equation*}
\begin{aligned}
g_M(n)
&  =  \frac{1}{M} \sum_{k = \max{\left\{n-M,-M\right\}}}^{\min{\left\{n+M,M\right\}}} \left( 1-\frac{|k|}{M} \right)\left( 1- \frac{|n-k|}{M}\right).
\end{aligned}
\end{equation*}
This particular choice of weights ensures that
\begin{equation*}
\begin{aligned}
\sum_{n = -2M}^{2M}\frac{1}{w_n^2}e^{i2\pi n\tau} &  = \frac{1}{M} \sum_{n= -2M}^{2M} g_M(n) e^{i2\pi n\tau} \\
& = \left[\frac{\sin (\pi M\tau)}{M \sin(\pi \tau)}\right]^4 \\
& =: K_M(\tau),
\end{aligned}
\end{equation*}
where $K_M(\tau)$ is known as the squared Fej{\'e}r kernel. Consequently, $\bK_M(\tau), \bK_M'(\tau) \in \complex^{K\times K}$ are the random matrix kernels produced from the squared Fej{\'e}r kernel. They have the form
\begin{equation*}
\begin{aligned}
\bK_M(\tau)
& =  \frac{1}{M} \sum_{n = -2M}^{2M} g_M(n)e^{i2\pi n\tau}  \bb_n \bb_n^H
\end{aligned}
\end{equation*}
and
\begin{equation*}
\begin{aligned}
\bK_M'(\tau)
& =  \frac{1}{M} \sum_{n = -2M}^{2M} (i2\pi n) g_M(n)e^{i2\pi n\tau} \bb_n \bb_n^H.
\end{aligned}
\end{equation*}
Denoting $\bK_M^{\ell}(\tau)$ as the $\ell$th entry-wise derivative of $\bK_M(\tau)$, we have
\begin{equation*}
\begin{aligned}
\expval \bK_M^{\ell}(\tau) & = \frac{1}{M} \sum_{n = -2M}^{2M} (i2\pi n)^{\ell} g_M(n)e^{i2\pi n\tau} \expval\left[ \bb_n \bb_n^H \right] \\
& = K_M^{\ell}(\tau)\bI_K,
\end{aligned}
\end{equation*}
where the second equation comes from the isotropy property of $\bb_n$.

We define
\begin{equation}
\label{eq:MMVkernel}
\begin{aligned}
\overline{\bq}^{\ell}(\tau) & = \sum_{j = 1}^J \left[ \expval \bK_M^{\ell}(\tau-\tau_j) \right] \overline{\balpha}_j + \sum_{j = 1}^J \left[ \expval \bK_M^{\ell+1} (\tau-\tau_j)\right] \overline{\bbeta}_j \\
&  = \sum_{j = 1}^J K_M^{\ell}(\tau-\tau_j) \overline{\balpha}_j + \sum_{j = 1}^J K_M^{\ell+1}(\tau-\tau_j) \overline{\bbeta}_j,
\end{aligned}
\end{equation}
where $\overline{\balpha}_j$ and $\overline{\bbeta}_j$ are solutions of the following linear system of equations:
\begin{equation*}
\begin{aligned}
\overline{\bq}(\tau_s) & = \sum_{j =1}^J  K_M(\tau_s -\tau_j) \overline{\balpha}_j + \sum_{j =1}^J  K_M'(\tau_s-\tau_j) \overline{\bbeta}_j \\
& = \sign(c_s) \bh_s,~~~~\tau_s \in \mathbb{D},
\end{aligned}
\end{equation*}
\begin{equation*}
\begin{aligned}
- \overline{\bq}'(\tau_s) & = -\left(\sum_{j =1}^J  K'_M(\tau_s -\tau_j) \overline{\balpha}_j + \sum_{j =1}^J  K_M''(\tau_s-\tau_j) \overline{\bbeta}_j\right) \\
& = \bzero_{K\times 1},~~~~\tau_s\in \mathbb{D}.
\end{aligned}
\end{equation*}
We note that $\overline{\bq}(\tau)$ is the vector dual polynomial constructed in  \cite{yang2014exact} for the full data multiple measurement vectors problem. Thus, the vector dual polynomial $\bq(\tau)$ constructed in our work is a random version of $\overline{\bq}(\tau)$, with the randomness introduced by $\bb_n$.
\subsection{Showing (\ref{cond1})}
\label{subsec: subsection2}
To show (\ref{cond1}), our effort becomes to explicitly find $\balpha_j$ and $\bbeta_j$ such that (\ref{cond1}) holds. To accomplish this, for all $\tau_s \in \mathbb{D}$, we plug the form of $\bq(\tau)$ back into (\ref{subcond1}) and (\ref{subcond2}). Thus, we can write
\begin{equation*}
\begin{aligned}
\bq(\tau_s) & = \sum_{j =1}^J  \bK_M(\tau_s -\tau_j) \balpha_j + \sum_{j =1}^J  \bK_M'(\tau_s-\tau_j) \bbeta_j \\
& = \sign(c_s) \bh_s,\\
- \bq'(\tau_s) & = - \left(\sum_{j =1}^J \bK_M'(\tau_s -\tau_j) \balpha_j  + \sum_{j =1}^J  \bK_M''(\tau_s-\tau_j) \bbeta_j\right)\\
& = \bzero_{K\times 1}.\\
\end{aligned}
\end{equation*}
Equivalently, we have the following linear system of equations
\begin{equation*}
\begin{aligned}
& \begin{bmatrix} \bD_0 & \frac{1}{\sqrt{|K_M^{''}(0)}|} \bD_1\\
-\frac{1}{\sqrt{|K_M^{''}(0)|}} \bD_1 & -\frac{1}{|K_M^{''}(0)|}\bD_2\end{bmatrix}\begin{bmatrix}\balpha \\ \sqrt{|K_M^{''}(0)|}\bbeta \end{bmatrix} = \begin{bmatrix} \bh \\ \bzero_{JK\times 1}\end{bmatrix},
\end{aligned}
\end{equation*}
where $\bh = \begin{bmatrix} (\sign(c_1)\bh_1)^H &  \dots & (\sign(c_J)\bh_J)^H \end{bmatrix}^H$,
$[\bD_{\ell}]_{sj} = \bK_M^{\ell}(\tau_s-\tau_j)$, and $K_M^{''}(0) = -\frac{4\pi^2 (M^2-1)}{3}$.

Denote
\begin{equation*}
\begin{aligned}
\bD & = \begin{bmatrix} \bD_0 & \frac{1}{\sqrt{|K_M^{''}(0)}|} \bD_1\\
-\frac{1}{\sqrt{|K_M^{''}(0)|}} \bD_1 & -\frac{1}{|K_M^{''}(0)|}\bD_2\end{bmatrix} \in \mathbb{C}^{2JK \times 2JK}.
\end{aligned}
\end{equation*}
Thus, as long as $\bD$ is invertible, one can obtain $\balpha_j$ and $\bbeta_j$.
To show the invertibility of $\bD$, we first show that $\expval \bD$ is invertible under the minimum separation condition in Lemma \ref{le:lemma1}. Then, in Lemma \ref{lem:lemma5.4}, we verify that $\bD$ is invertible with high probability after arguing that $\bD$ is close to $\expval \bD$ given enough measurements in Lemma \ref{lem:matrixbernstein}. Defining
\begin{equation*}
\bE(n) = \begin{bmatrix}
e^{i2\pi n \tau_1}\\
\vdots \\
e^{i2\pi n \tau_J}\\
\frac{-i2\pi n}{\sqrt{|K_M^{''}(0)|}} e^{i2\pi n\tau_1}\\
\vdots \\
\frac{-i2\pi n}{\sqrt{|K_M^{''}(0)|}} e^{i2\pi n \tau_J}
\end{bmatrix},
\end{equation*}
we have
\begin{equation*}
\bD  = \frac{1}{M} \sum_{n = -2M}^{2M} \left( \left( g_M(n) \bE(n) \bE(n)^H \right) \otimes \left(\bb_n \bb_n^H\right) \right)
\end{equation*}
and
\begin{equation*}
\begin{aligned}
\mathbb{E} \bD & = \frac{1}{M} \sum_{n = -2M}^{2M} \left( \left( g_M(n) \bE(n) \bE(n)^H \right) \otimes \bI_K \right) \\
& =  \bD'\otimes \bI_K,
\end{aligned}
\end{equation*}
where
\begin{equation*}
\begin{aligned}
\bD' & = \begin{bmatrix} \bD'_0 & \frac{1}{\sqrt{|K_M^{''}(0)}|} \bD'_1\\
-\frac{1}{\sqrt{|K_M^{''}(0)|}} \bD'_1 & -\frac{1}{|K_M^{''}(0)|}\bD'_2\end{bmatrix}\in \mathbb{C}^{2J \times 2J}
\end{aligned}
\end{equation*}
with $[\bD'_{\ell}]_{sj} = K_M^{\ell}(\tau_s-\tau_j)$.
\begin{lemma}
\label{le:lemma1}
Suppose that $\Delta_{\tau} \geq \frac{1}{M}$. Then, $\expval \bD$ is invertible and
\begin{equation*}
\norm[]{\bI_{2JK} - \expval \bD} \leq 0.3623,
\end{equation*}
\begin{equation*}
\norm[]{\expval \bD} \leq 1.3623,
\end{equation*}
\begin{equation*}
\norm[]{\left(\expval \bD\right)^{-1}} \leq 1.568.
\end{equation*}
\end{lemma}
The proof of Lemma \ref{le:lemma1} is given in Appendix \ref{sec:prooflemma1}. We denote $\bD'^{-1} = \begin{bmatrix} \bL' & \bR'\end{bmatrix}$, where $\bL'\in \mathbb{C}^{2J\times J}$ and $\bR' \in \mathbb{C}^{2J\times J}$.
To have a concentration of measure result for $\bD$ as given in the following lemma, we note that
\begin{equation*}
\begin{aligned}
&  \bD - \mathbb{E} \bD =  \frac{1}{M} \sum_{n = -2M}^{2M} \left( \left( g_M(n) \bE(n) \bE(n)^H \right) \otimes \left( \bb_n \bb_n^H - \bI_{K}\right)\right),
\end{aligned}
\end{equation*}
which is a sum of independent matrices of zero mean. It is derived in Lemma 2 of  \cite{chi2015guaranteed}.
\begin{lemma}
\label{lem:matrixbernstein}
 \cite{chi2015guaranteed} For any $\varepsilon_1 \in (0,0.6376), 0<\delta_1<1$, if
$M \geq \frac{80 \mu JK}{\varepsilon_1^2} \log\left( \frac{4JK}{\delta_1}\right)$ and $\Delta_{\tau}\geq \frac{1}{M}$,
then $\norm[]{\bD-\expval \bD} \leq \varepsilon_1$ with probability at least $1-\delta_1$.
\end{lemma}
\begin{lemma}
\label{lem:lemma5.4}
Under the assumptions of Lemma \ref{lem:matrixbernstein}, $\bD$ is invertible with probability at least $1-\delta_1$.
\end{lemma}
\begin{proof}
\begin{equation*}
\begin{aligned}
\norm[]{\bI_{2JK} - \bD}  & \leq \norm[]{\bI_{2JK} - \expval \bD} + \norm[]{\expval \bD - \bD} \\
& \leq 0.3623 + \varepsilon_1 \\
& <1,
\end{aligned}
\end{equation*}
where the first inequality uses the triangle inequality, and the second inequality follows from Lemmas \ref{le:lemma1} and \ref{lem:matrixbernstein}.
\end{proof}
For $\varepsilon_1>0$, we define the event $\cE_{1, \varepsilon_1} = \left\{ \norm[]{\bD-\expval \bD} \leq \varepsilon_1\right\}$. The following lemma will also be useful later in our analysis.
\begin{lemma}
\label{lem:inverseopnorm}
\cite{chi2015guaranteed,tang2013compressed} Under the event $\cE_{1, \varepsilon_1}$ with $\varepsilon_1 \in \left(0,\frac{1}{4}\right]$, we have
\begin{equation*}
\norm[]{\bD^{-1} - (\expval \bD)^{-1}} \leq 2 \norm[]{(\expval \bD)^{-1}}^2 \varepsilon_1,
\end{equation*}
and
\begin{equation*}
\norm[]{\bD^{-1}} \leq 2 \norm[]{(\expval \bD)^{-1}}.
\end{equation*}
\end{lemma}

Therefore, the construction of the dual polynomial ensures that condition (\ref{cond1}) is satisfied automatically. It should also be pointed out that Lemma \ref{lem:lemma5.4} essentially guarantees that the set formed in the second condition of Proposition \ref{optimalitycond} is linearly independent. In the following subsection, we focus on showing (\ref{cond2}):  $\norm[2]{\bq(\tau)}<1$ for all $\tau \notin \mathbb{D}$. The proof is partitioned into the following three major steps:
\begin{itemize}
\item Showing that the random dual polynomial $\bq(\tau)$ concentrates around $\overline{\bq}(\tau)$ on a discrete set $\Omega_{\Grid}$.
\item Then, showing that the random dual polynomial $\bq(\tau)$ concentrates around $\overline{\bq}(\tau)$ everywhere in $[0,1)$.
\item Eventually, showing that $\norm[]{\bq(\tau)}<1, \tau \notin \mathbb{D}$.
\end{itemize}
This proof strategy was first developed in \cite{tang2013compressed} for compressed sensing off the grid, and was later adopted by  \cite{li2014off,yang2014exact} for line spectrum estimation with multiple measurement vectors and by  \cite{chi2015guaranteed} for blind sparse spikes deconvolution.

\subsection{Showing (\ref{cond2})}
\label{subsec:subsection3}
We first show that $\bq^{\ell}(\tau)$ concentrates around $\overline{\bq}^{\ell} (\tau)$ on a finite discrete set $\Omega_{\Grid}$, whose size $|\Omega_{\Grid}|$ will be determined later on. For this purpose, we introduce
\begin{align*}
\bV_{\ell}(\tau)
 & = \frac{1}{\sqrt{|K''_M(0)|}^{\ell}}
\begin{bmatrix}
\bK_M^{\ell} (\tau-\tau_1)^{H} \\
\vdots \\
\bK_M^{\ell} (\tau-\tau_J)^{H} \\
\frac{1}{\sqrt{|K''_M(0)|}}\bK_M^{\ell +1}(\tau-\tau_1)^{H}\\
\vdots \\
\frac{1}{\sqrt{|K''_M(0)|}}\bK_M^{\ell +1}(\tau-\tau_J)^{H}\\
\end{bmatrix} \notag \displaybreak\\
& = \frac{1}{M} \sum_{n = -2M}^{2M}  g_M(n)\left(\frac{-i2\pi n}{\sqrt{|K''_M(0)|}}\right)^{\ell} e^{-i2\pi  n \tau}\bE(n)\otimes \bb_n \bb_n^H \in \complex^{2JK\times K}.
\end{align*}
Taking the expectation of the above leads to
\begin{equation*}
\begin{aligned}
 \expval \bV_{\ell} (\tau)
&  = \frac{1}{M} \sum_{n = -2M}^{2M} g_M(n)\left(\frac{-i2\pi n}{\sqrt{|K''_M(0)|}}\right)^{\ell} e^{-i2\pi n  \tau} \bE(n) \otimes \bI_{K} \\
 &=  \frac{1}{\sqrt{|K''_M(0)|}^{\ell}}
\begin{bmatrix}
K_M^{\ell} (\tau-\tau_1)^{*} \\
\vdots \\
K_M^{\ell} (\tau-\tau_J)^{*} \\
\frac{1}{\sqrt{|K''_M(0)|}}K_M^{\ell +1}(\tau-\tau_1)^{*}\\
\vdots \\
\frac{1}{\sqrt{|K''_M(0)|}}K_M^{\ell +1}(\tau-\tau_J)^{*}\\
\end{bmatrix} \otimes \bI_{K}\\
& =: \bv_{\ell}(\tau) \otimes \bI_K.
\end{aligned}
\end{equation*}
Setting $\bD^{-1} = \begin{bmatrix} \bL & \bR \end{bmatrix}$, with $\bL, \bR\in \complex^{2JK\times JK}$
and using the fact that
\begin{equation*}
\bD \begin{bmatrix} \balpha \\  \sqrt{|K''_M(0)|}\bbeta \end{bmatrix}  = \begin{bmatrix} \bh \\ \bzero_{JK\times 1}\end{bmatrix},
\end{equation*}
we have
\begin{equation*}
\begin{aligned}
\frac{1}{\sqrt{|K''_M(0)|}^{\ell}} \bq^{\ell} (\tau) 
 &=  \sum_{j = 1}^J\frac{1}{\sqrt{|K''_M(0)|}^{\ell}} \bK_M^{\ell} (\tau- \tau_j) \balpha_j + \sum_{j =1}^J
\sqrt{|K''_M(0)|} \frac{1}{\sqrt{|K''_M(0)|}^{\ell + 1}} \bK_{M}^{\ell +1} (\tau- \tau_j)\bbeta_j \\
 &= \bV_{\ell}(\tau)^H \bD^{-1} \begin{bmatrix} \bh \\ \bzero_{JK \times 1}\end{bmatrix}\\
 & =  \bV_{\ell}(\tau)^H \bL \bh.
\end{aligned}
\end{equation*}
We decompose $\bV_{\ell}(\tau)^H \bL \bh$ into three parts as follows:
\begin{equation*}
\begin{aligned}
\bV_{\ell}(\tau)^H \bL \bh & = \left( \bV_{\ell}(\tau) - \mathbb{E} \bV_{\ell}(\tau) + \mathbb{E} \bV_{\ell}(\tau) \right)^H \left( \bL - \bL' \otimes \bI_{K} + \bL' \otimes \bI_{K} \right) \bh\\
 & =  \left[\mathbb{E}\bV_{\ell}(\tau)\right]^H \left( \bL'\otimes \bI_{K} \right) \bh
+  (\bV_{\ell}(\tau) - \mathbb{E} \bV_{\ell}(\tau) )^H \bL \bh + \left[\mathbb{E} \bV_{\ell}(\tau)\right]^H  \left(\bL - \bL' \otimes \bI_{K}\right) \bh\\
  & = \left[\mathbb{E}\bV_{\ell}(\tau)\right]^H \left( \bL'\otimes \bI_{K} \right) \bh + \bI_1^{\ell}(\tau) + \bI_2^{\ell}(\tau),
\end{aligned}
\end{equation*}
where we have defined
\begin{equation*}
\bI_1^{\ell}(\tau) = (\bV_{\ell}(\tau) - \mathbb{E} \bV_{\ell}(\tau) )^H \bL \bh,\end{equation*}
\begin{equation*}
\bI_2^{\ell}(\tau)= \left[ \mathbb{E} \bV_{\ell}(\tau)\right]^H  \left(\bL - \bL' \otimes \bI_{K}\right) \bh.
\end{equation*}
For the first term $\left[\mathbb{E}\bV_{\ell}(\tau)\right]^H \left( \bL'\otimes \bI_{K} \right) \bh$ above, we notice that
\begin{equation}
\label{eq:predecomposition}
\begin{aligned}
\left[\mathbb{E}\bV_{\ell}(\tau)\right]^H \left( \bL'\otimes \bI_{K} \right) \bh   
  &= (\bv_{\ell}(\tau) \otimes \bI_K)^H (\bL'\otimes \bI_K)\bh \\
  &= (\bv_{\ell}(\tau) \otimes \bI_K)^H \begin{bmatrix} \bL'\otimes \bI_K & \bR' \otimes \bI_K \end{bmatrix} \begin{bmatrix} \bh \\ \bzero_{JK\times 1}\end{bmatrix}\\
 &= (\bv_{\ell}(\tau) \otimes \bI_K)^H \left(\bD'\otimes \bI_K\right)^{-1} \begin{bmatrix} \bh \\ \bzero_{JK\times 1}\end{bmatrix}\\
 &= (\bv_{\ell}(\tau) \otimes \bI_K)^H \begin{bmatrix}\overline{\balpha} \\
\sqrt{|K''_M(0)|}\overline{\bbeta} \end{bmatrix},
\end{aligned}
\end{equation}
where $\overline{\balpha} = \begin{bmatrix}
\overline{\balpha}_1^H &
\cdots &
\overline{\balpha}_J^H \end{bmatrix}^H$ and $\overline{\bbeta} = \begin{bmatrix}
\overline{\bbeta}_1^H &
\cdots &
\overline{\bbeta}_J^H \end{bmatrix}^H.$
Matching (\ref{eq:predecomposition}) to (\ref{eq:MMVkernel}), we have
\begin{equation*}
\left[\mathbb{E}\bV_{\ell}(\tau)\right]^H \left( \bL'\otimes \bI_{K} \right) \bh  =\frac{1}{\sqrt{|K''_M(0)|}^{\ell}} \overline{\bq}^{\ell} (\tau).
\end{equation*}
As a consequence, we obtain the following decomposition for $\frac{1}{\sqrt{|K''_M(0)|}^{\ell}} \bq^{\ell} (\tau)$:
\begin{equation}
\label{eq:decomposition}
\begin{aligned}
\frac{1}{\sqrt{|K''_M(0)|}^{\ell}} \bq^{\ell} (\tau)  
 &= \frac{1}{\sqrt{|K''_M(0)|}^{\ell}} \overline{\bq}^{\ell} (\tau) + \bI_1^{\ell}(\tau) + \bI_2^{\ell}(\tau).
\end{aligned}
\end{equation}
We now present the roadmap for the rest of the proof. First, we demonstrate that $\norm[2]{\bI_1^{\ell}(\tau)}$ and $\norm[2]{\bI_2^{\ell}(\tau)}$ are small on the set of grid points $\Omega_{\Grid}$ in Sections \ref{subsubsec: boundfirstterm} and \ref{subsubsec: boundsecondterm}, respectively. Section \ref{subsubsec: bounddeviationongridpoint} combines these facts to establish that $\frac{1}{\sqrt{|K''_M(0)|}^{\ell}} \bq^{\ell} (\tau)$ is close to $\frac{1}{\sqrt{|K''_M(0)|}^{\ell}} \overline{\bq}^{\ell} (\tau)$ on $\Omega_{\Grid}$. Finally, in Section \ref{subsubsec: showeverywhere}, we extend the argument in Section \ref{subsubsec: bounddeviationongridpoint} to the whole continuous domain $[0,1)$ and show that $\norm[2]{\bq(\tau)}<1, \tau \in [0,1)\setminus\mathbb{D}$. \\

\subsubsection{Bound $\norm[2]{\bI_1^{\ell}(\tau)}$ on the set of points $\Omega_{\Grid}$}
\label{subsubsec: boundfirstterm}

In order to bound $\norm[2]{\bI_1^{\ell}(\tau)}$, we will apply the matrix Bernstein inequality by exploiting the randomness of $\bh_j$. For this purpose, we first need to control $\norm[]{(\bV_{\ell}(\tau) - \expval \bV_{\ell}(\tau))^H \bL}$, which further requires estimating $\norm[]{\bV_{\ell}(\tau) - \expval \bV_{\ell}(\tau)}$. The latter is established in Lemma \ref{lem:matroperatornormbound}, whose
proof uses the matrix Bernstein inequality and is provided in Appendix \ref{operatornormboundonV}.
\begin{lemma}
\label{lem:matroperatornormbound}
Fix $\tau \in [0,1)$, and let $0<\varepsilon_2<1$. Then, for $\ell = 0, 1, 2, 3$, $\norm[]{\bV_{\ell}(\tau) - \expval \bV_{\ell}(\tau)} \leq \varepsilon_2$ holds with probability at least $1-4\delta_2$ provided that $M \geq \frac{640 \cdot 4^{2 \ell} \mu JK} {3 \varepsilon_2^2}\log\left(\frac{2JK +K}{ \delta_2}\right).$
\end{lemma}
We define the event $\cE_{2,\varepsilon_2} = \left\{ \norm[]{\bV_{\ell}(\tau) - \expval \bV_{\ell}(\tau)} \leq \varepsilon_2, \ell = 0,1,2,3\right\}$. The following two lemmas control $\norm[]{\left(\bV_{\ell}(\tau) - \expval \bV_{\ell}(\tau)\right)^H \bL}$ and $\norm[F]{\bV_{\ell}(\tau) - \expval \bV_{\ell}(\tau)}$, respectively.
\begin{lemma}
\label{lem:unionfailprob}
Assume that $\varepsilon_1 \in (0,\frac{1}{4}]$. Consider the finite set of grid points $\Omega_{\Grid} = \{\tau_d \}$, whose cardinality $|\Omega_{\Grid}|$ will be determined later. We have
\begin{equation*}
\begin{aligned}
\mathbb{P} \left\{ \sup_{\tau_d \in \Omega_{\Grid}} \norm[]{\left(\bV_{\ell}(\tau_d) - \expval \bV_{\ell}(\tau_d)\right)^H \bL} \geq 4 \varepsilon_2,
\vphantom{\sup_{\tau_d \in \Omega_{\Grid}} \norm[]{\left(\bV_{\ell}(\tau_d) - \expval \bV_{\ell}(\tau_d)\right)^H \bL}} \ell = 0, 1, 2, 3 \right\} \leq|\Omega_{\Grid}|4 \delta_2 + \mathbb{P}\left(\cE_{1, \varepsilon_1}^c\right)
\end{aligned}
\end{equation*}
provided that $M \geq \frac{640 \cdot 4^{2\ell} \mu JK} {3 \varepsilon_2^2}\log\left(\frac{2JK +K}{ \delta_2}\right)$.
\end{lemma}
Lemma \ref{lem:unionfailprob} is a consequence of Lemmas \ref{lem:inverseopnorm} and \ref{lem:matroperatornormbound}. Its proof is given in Appendix \ref{proofofLemma5_6}.
\begin{lemma}
\label{lem:froboperatornormineq}
Conditioned on the event $\cE_{2,\varepsilon_2}$, we have
\begin{equation*}
\begin{aligned}
\norm[F]{\bV_{\ell}(\tau) - \expval \bV_{\ell}(\tau)}  & \leq \sqrt{K} \norm[]{\bV_{\ell} (\tau) - \expval \bV_{\ell}(\tau)}
\leq \sqrt{K} \varepsilon_2.
\end{aligned}
\end{equation*}
\end{lemma}

Based on Lemmas \ref{lem:unionfailprob} and \ref{lem:froboperatornormineq},
the following lemma shows that $\norm[2]{\bI_1^{\ell}(\tau)}$ can be well-controlled on the set of points $\Omega_{\Grid}$ with high probability.
\begin{lemma}
\label{lem:vectorbernstein}
Assume that $\bh_j \in \complex^K, j = 1, \dots, J$, are i.i.d.\ symmetric random samples from the complex unit sphere $\mathbb{CS}^{K-1}$. There exists a numerical constant $C$ such that if
\begin{equation*}
\begin{aligned}
M \geq C\mu JK \max\left\{ \frac{1}{\varepsilon_4^2} \log\left(\frac{|\Omega_{\Grid}|JK}{\delta}\right)
 \vphantom{\frac{1}{\varepsilon_4^2} \log\left(\frac{|\Omega_{\Grid}|JK}{\delta}\right) \log^2 \left(\frac{|\Omega_{\Grid}|K}{\delta}\right)}   \log^2 \left(\frac{|\Omega_{\Grid}|K}{\delta}\right), \log\left(\frac{JK}{\delta}\right)\right\},
\end{aligned}
\end{equation*}
then we have
\begin{equation*}
\begin{aligned}
\mathbb{P} \left\{ \sup_{\tau_d \in \Omega_{\Grid}}\norm[2]{\bI_1^{\ell}(\tau_d)}\leq \varepsilon_4, \ell = 0,1,2, 3\right\} \geq 1-12 \delta.\\
\end{aligned}
\end{equation*}
\end{lemma}
We provide the detailed proof in Appendix \ref{proofofLemma5_8}, which is based on an application of the matrix Bernstein inequality for the vector case. \\

\subsubsection{Bound $\norm[2]{\bI_2^{\ell}(\tau)}$ on the set of points $\Omega_{\Grid}$}
\label{subsubsec: boundsecondterm}
The strategy for bounding $\norm[2]{\bI_2^{\ell}(\tau)}$ is similar to the one in Section \ref{subsubsec: boundfirstterm}. First of all, we use the following lemma to bound $\norm[F]{\left(\bL -\bL'\otimes \bI_K \right)^H\expval \bV_{\ell}(\tau)}^2$. Its proof is given in Appendix \ref{appendix:lemma5.9}.
\begin{lemma}
\label{lem:frobeniusnorm2}
Conditioned on the event $\cE_{1, \varepsilon_1}$ with $ \varepsilon_1 \in (0,\frac{1}{4}]$, we have
\begin{equation*}
\norm[F]{\left(\bL -\bL'\otimes \bI_K \right)^H\expval \bV_{\ell}(\tau)}^2 \leq CK\varepsilon_1^2
\end{equation*}
for some constant $C$.
\end{lemma}
Lemma \ref{lem:frobeniusnorm2} together with the matrix Bernstein inequality ensure that $\norm[2]{\bI_2^{\ell}(\tau)}$ is small on the set of points $\Omega_{\Grid}$ with high probability.
\begin{lemma}
\label{lem:boundI2}
Assume that each $\bh_j \in \complex^K$ is an i.i.d.\ symmetric random sample from the complex unit sphere $\mathbb{CS}^{K-1}$. There exists a numerical constant $C$ such that if
\begin{equation*}
\begin{aligned}
M \geq C \frac{\mu JK}{\varepsilon_5^2} \log\left(\frac{JK}{\delta}\right)\log^2\left(\frac{|\Omega_{\Grid}|K}{\delta}\right),
\end{aligned}
\end{equation*}
we have
\begin{equation*}
\begin{aligned}
\mathbb{P} \left\{ \sup_{\tau_d \in \Omega_{\Grid}}\norm[2]{\bI_2^{\ell}(\tau_d)}\leq \varepsilon_5, \ell = 0,1,2,3\right\} \geq 1- 8 \delta.\\
\end{aligned}
\end{equation*}
\end{lemma}
The proof of Lemma \ref{lem:boundI2} is provided in Appendix \ref{appendix:lemma5.10}. \\

\subsubsection{Show that $\frac{1}{\sqrt{|K''_M(0)|}^{\ell}} \bq^{\ell} (\tau)$ is close to $ \frac{1}{\sqrt{|K''_M(0)|}^{\ell}} \overline{\bq}^{\ell} (\tau)$ on the set of points $\Omega_{\Grid}$}
\label{subsubsec: bounddeviationongridpoint}
Define the event
\begin{equation*}
\begin{aligned}
\cE = \left\{\sup_{\tau_d \in \Omega_{\Grid}}\norm[2]{\frac{1}{\sqrt{|K''_M(0)|}^{\ell}} \bq^{\ell} (\tau_d) -  \frac{1}{\sqrt{|K''_M(0)|}^{\ell}} \overline{\bq}^{\ell} (\tau_d)}
\vphantom{\sup_{\tau_d \in \Omega_{\Grid}}\norm[2]{\frac{1}{\sqrt{|K''_M(0)|}^{\ell}} \bq^{\ell} (\tau_d) -  \frac{1}{\sqrt{|K''_M(0)|}^{\ell}} \overline{\bq}^{\ell} (\tau_d)}} \leq \frac{\varepsilon}{3}, \ell = 0,1,2,3\right\}.
 \end{aligned}
\end{equation*}

Lemmas \ref{lem:vectorbernstein} and \ref{lem:boundI2} along with the decomposition (\ref{eq:decomposition}) immediately result in the following proposition.
\begin{prop}
\label{prop:concentrationeq}
Assume that $\Omega_{\Grid} \subset [0,1)$ is a finite set of points. Let $\delta\in (0,1)$ be the failure probability. Then, there exists a numerical constant $C$ such that when
\begin{equation*}
\begin{aligned}
M \geq C\mu JK \max \left\{ \frac{1}{\varepsilon^2}\log\left(\frac{|\Omega_{\Grid}|JK}{\delta}\right)
\vphantom{\frac{1}{\varepsilon^2}\log\left(\frac{|\Omega_{\Grid}|JK}{\delta}\right) \log^2\left(\frac{|\Omega_{\Grid}|K}{\delta}\right)} \log^2\left(\frac{|\Omega_{\Grid}|K}{\delta}\right),\log\left(\frac{JK}{\delta} \right)\right\},
\end{aligned}
\end{equation*}
we have
\begin{equation*}
\mathbb{P} (\cE) \geq 1-\delta.
\end{equation*}
\end{prop}

\subsubsection{Show that $\norm[2]{\bq (\tau)} <1$ everywhere,  $\tau \in [0,1) \setminus \mathbb{D}$}
\label{subsubsec: showeverywhere}
The following lemma shows that $\frac{1}{\sqrt{|K''_M(0)|}^{\ell}} \bq^{\ell} (\tau)$ is close to $ \frac{1}{\sqrt{|K''_M(0)|}^{\ell}} \overline{\bq}^{\ell} (\tau)$ everywhere in $[0,1)$. The proof, which involves Bernstein's polynomial inequality, is given in Appendix \ref{proofofLemma5_11}.
\begin{lemma}
\label{lem:continuousextension}
Assume that $\Delta_{\tau} \geq \frac{1}{M}$. Then for all $\tau\in [0,1)$ and $\ell = 0,1,2,3$, the following holds with probability at least $1-\delta$
\begin{equation*}
 \norm[2]{\frac{1}{\sqrt{|K''_M(0)|}^{\ell}} \bq^{\ell} (\tau) -   \frac{1}{\sqrt{|K''_M(0)|}^{\ell}} \overline{\bq}^{\ell} (\tau)} \leq \varepsilon
\end{equation*}
provided that
$M \geq C\mu JK \max \left\{ \frac{1}{\varepsilon^2}\log\left(\frac{MJK}{\varepsilon \delta}\right) \log^2\left(\frac{MK}{\varepsilon \delta}\right), \log\left(\frac{JK}{\delta} \right)\right\}$ for some numerical constant $C$.
\end{lemma}
Next, we show that $\norm[2]{\bq(\tau)}<1$ everywhere, $\tau \in [0,1)\setminus\mathbb{D}$.
To do this, define
\begin{equation*}
\Omega_{\near} = \bigcup_{j = 1}^{J} [\tau_j-\tau_{b,1}, \tau_j + \tau_{b,1}],
\end{equation*}
\begin{equation*}
\Omega_{\far} = [0,1) \setminus \Omega_{\near}
\end{equation*}
with $\tau_{b,1} = 8.245\times 10^{-2} \frac{1}{M}$.
An immediate consequence of Lemma \ref{lem:continuousextension} is the following result, which verifies that $\norm[2]{\bq(\tau)} <1, \forall~\tau\in \Omega_{\far}$.
\begin{lemma}
\label{lem:farregionproof}
Assume that $\Delta_{\tau} \geq \frac{1}{M}$ and that
\begin{equation*}
M \geq C\mu JK\log\left(\frac{MJK}{ \delta}\right) \log^2\left(\frac{MK}{ \delta}\right),
\end{equation*}
for some positive numerical constant $C$.
Then
\begin{equation*}
\norm[2]{\bq (\tau)} <1, ~~~~\forall~\tau \in \Omega_{\far}
\end{equation*}
with probability at least $1-\delta$.
\end{lemma}
\begin{proof} Taking $\varepsilon = 10^{-5}$ in Lemma \ref{lem:continuousextension}, we obtain that
\begin{equation*}
\begin{aligned}
\norm[2]{\bq (\tau)} & \leq \norm[2]{\bq(\tau) - \overline{\bq}(\tau) }+\norm[2]{\overline{\bq} (\tau)} \\
& \leq 10^{-5} + \norm[2]{\overline{\bq} (\tau)}.
\end{aligned}
\end{equation*}
In order to bound $\norm[2]{\bq(\tau)}$, it remains to control $\norm[2]{\overline{\bq}(\tau)} $ for $\tau \in \Omega_{\far}$. From (\ref{eq:predecomposition}), note that
\begin{equation*}
\begin{aligned}
\norm[2]{\overline{\bq} (\tau)} & = \sup_{\bu: \norm[2]{\bu} = 1} \bu^H \left(\left[\expval \bV(\tau)\right]^H  (\bL' \otimes \bI_K) \bh \right) \\
& = \sup_{\bu: \norm[2]{\bu} = 1} \bu^H \left(( \bv(\tau) \otimes \bI_K)^H (\bL' \otimes \bI_K) \bh \right)\\
& = \sup_{\bu: \norm[2]{\bu} = 1}\sum_{j=1}^J  \left(\bv (\tau)^H \bL'\right)(j) (\bu^H \sign(c_j) \bh_j )\\
& \leq 0.99992,
\end{aligned}
\end{equation*}
where for the third line above we have, with some abuse of notation, denoted $\left(\bv(\tau)^H \bL'\right)(j) $ as the $j$th entry of the row vector $\bv(\tau)^H \bL' $. The fourth line follows from the fact that $|\bu^H \sign(c_j) \bh_j|\leq 1$ and from the proof of Lemma 2.4 in \cite{candes2014towards} for $\tau\in \Omega_{\far}$. Thus, we have shown
\begin{equation*}
\norm[2]{\bq(\tau)}<1, ~~~~\forall~\tau \in \Omega_{\far}.
\end{equation*}
\end{proof}
The next lemma shows that $\norm[2]{\bq(\tau)} <1$ when $\tau\in \Omega_{\near}$.
\begin{lemma}
\label{lem:nearregionproof}
Assume that $\tau\in \Omega_{\near}$. Then as long as
\begin{equation*}
M \geq C\mu JK  \log\left(\frac{MJK}{ \delta}\right) \log^2\left(\frac{MK}{ \delta}\right),
\end{equation*}
we have
\begin{equation*}
\norm[2]{\bq(\tau)}<1, ~~~~\forall~\tau \in \Omega_{\near}
\end{equation*}
with probability at least $1-\delta$.
\end{lemma}
\begin{proof} First of all, for $\tau_j \in \mathbb{D}$, we have $\frac{d\norm[2]{\bq(\tau)}^2}{d \tau}|_{\tau = \tau_j} = 2 \langle\bq'(\tau_j), \bq(\tau_j) \rangle_{\mathbb{R}}  = 0$. The Taylor's expansion of $\norm[2]{\bq(\tau)}^2$ in the interval $\left[ \tau_j-\tau_{b,1}, \tau_j + \tau_{b,1}\right]$ is given by 
\begin{equation*}
\begin{aligned}
\norm[2]{\bq(\tau)}^2 & = \norm[2]{\bq(\tau_j)}^2 + \frac{d\norm[2]{\bq(\tau)}^2}{d \tau}|_{\tau = \tau_j} (\tau-\tau_j) + \frac{1}{2} \frac{d^2 \norm[2]{\bq(\tau)}^2}{d \tau^2}|_{\tau = \xi} (\tau-\tau_j)^2,
\end{aligned}
\end{equation*}
for some $\xi \in \left[ \tau_j-\tau_{b,1}, \tau_j + \tau_{b,1}\right]$.
This implies that in order to show $\norm[2]{\bq(\tau)}<1,\tau \in \Omega_{\near}$, it is sufficient to verify that
\begin{equation*}
\begin{aligned}
\frac{1}{2} \frac{d^2 \norm[2]{\bq(\tau)}^2}{d \tau^2}&  = \norm[2]{\bq'(\tau)}^2 +{\text {Re}}\left\{ \bq''(\tau)^H \bq(\tau) \right\}<0
\end{aligned}
\end{equation*}
for $\tau \in \Omega_{\near}$. Note that
\begin{equation*}
\begin{aligned}
 \frac{1}{|K''_M(0)| } \norm[2]{\bq'(\tau)}^2
   & = \norm[2]{\frac{1}{\sqrt{|K''_M(0)|} }\bq'(\tau)}^2 \\
  & = \norm[2]{\frac{1}{\sqrt{|K''_M(0)|}}\left(\bq'(\tau) - \overline{\bq}'(\tau) + \overline{\bq}'(\tau) \right)}^2 \\
 & \leq \varepsilon^2 + 2 \varepsilon \norm[2]{\frac{1}{\sqrt{|K''_M(0)|} } \overline{\bq}'(\tau)} + \norm[2]{\frac{1}{\sqrt{|K''_M(0)|} }\overline{\bq}'(\tau)}^2,
\end{aligned}
\end{equation*}
where the inequality above follows from Lemma \ref{lem:continuousextension}.
We also need the following estimate on $\norm[2]{\overline{\bq}'(\tau)}$:
\begin{equation*}
\begin{aligned}
\norm[2]{\overline{\bq}' (\tau)}  \leq 1.5765M
\end{aligned}
\end{equation*}
as given in Appendix \ref{sec:suppmaterial}.
Therefore, we have
\begin{equation*}
\begin{aligned}
 \frac{1}{|K''_M(0)| } \norm[2]{\bq'(\tau)}^2 
  & <  \varepsilon^2 + 1.7383 \varepsilon +  \norm[2]{\frac{1}{\sqrt{|K''_M(0)|} }\overline{\bq}'(\tau)}^2 \\
& =  \varepsilon^2 + 1.7383 \varepsilon + \frac{1}{|K_M''(0)|}\norm[2]{\overline{\bq}'(\tau)}^2,
\end{aligned}
\end{equation*}
where the inequality above follows from the fact that $\frac{1}{\sqrt{|K''_M(0)|} } < \frac{1}{M}\sqrt{\frac{3}{\pi^2}}$ for $M\geq 2$.

Next, we bound $\frac{1}{|K''_M(0)| }{\text {Re}}\left\{ \bq''(\tau)^H \bq(\tau) \right\}$:
\begin{equation*}
\begin{aligned}
& \frac{1}{|K''_M(0)| }{\text {Re}}\left\{ \bq''(\tau)^H \bq(\tau) \right\} \\
 & =  {\text {Re}} \left\{\left( \frac{1}{|K_M''(0)|}\left(\bq''(\tau) - \overline{\bq}''(\tau) \right)+ \frac{1}{|K_M''(0)|}\overline{\bq}''(\tau)\right)^H \vphantom{\left( \frac{1}{|K_M''(0)|}\left(\bq''(\tau) - \overline{\bq}''(\tau) \right)+ \frac{1}{|K_M''(0)|}\overline{\bq}''(\tau)\right)^H}\left(\bq(\tau) -\overline{\bq}(\tau) + \overline{\bq}(\tau)\right) \right\} \\
 & = {\text {Re}}\left\{\left( \frac{1}{|K_M''(0)|}\left(\bq''(\tau) - \overline{\bq}''(\tau) \right) \right)^H(\bq(\tau) - \overline{\bq}(\tau))\right\} + {\text {Re}}\left\{\left(\frac{1}{|K_M''(0)|}\overline{\bq}''(\tau)\right)^H \overline{\bq}(\tau) \right\} \\
 & ~~~~+  {\text {Re}}\left\{ \left( \frac{1}{|K_M''(0)|}\left(\bq''(\tau) - \overline{\bq}''(\tau) \right) \right)^H \overline{\bq}(\tau) \right\} + {\text {Re}}\left\{\left(\frac{1}{|K_M''(0)|}\overline{\bq}''(\tau)\right)^H (\bq(\tau)-\overline{\bq}(\tau))  \right\}\\
 &\leq \varepsilon^2 + 4.2498\varepsilon +  {\text {Re}}\left\{\left(\frac{1}{|K_M''(0)|}\overline{\bq}''(\tau)\right)^H \overline{\bq}(\tau) \right\},
\end{aligned}
\end{equation*}
where in the last inequality we have used Lemma \ref{lem:continuousextension}, the fact that $\norm[2]{\overline{\bq}(\tau)} \leq 1.0361$ for $\tau \in \Omega_{\near}$, and the fact that $\norm[2]{\overline{\bq}''(\tau)} \leq 21.1451M^2$  for $\tau \in \Omega_{\near}$ from Appendix \ref{sec:suppmaterial}.
Thus, we have
\begin{equation*}
\begin{aligned}
\frac{1}{|K_M''(0)|}\frac{1}{2} \frac{d^2 \norm[2]{\bq(\tau)}^2}{d \tau^2} & = \frac{1}{|K_M''(0)|}\left\{ \norm[2]{\bq'(\tau)}^2 +{\text {Re}}\left\{ \bq''(\tau)^H \bq(\tau) \right\} \right\}\\
 &<  2\varepsilon^2 + 5.9881 \varepsilon + \frac{1}{|K_M''(0)|} \left\{ \norm[2]{\overline{\bq}'(\tau)}^2 
\vphantom{ \norm[2]{\overline{\bq}'(\tau)}^2} + {\text {Re}}\left\{ \overline{ \bq}''(\tau)^H\overline{ \bq}(\tau) \right\} \right\} \\
 & \leq  2\varepsilon^2 + 5.9881 \varepsilon + \frac{1}{|K_M''(0)|} (-0.3756M^2) \\
 & <  2\varepsilon^2 + 5.9881  \varepsilon - 0.0285\\
& <  0,
\end{aligned}
\end{equation*}
where  in the fourth line above we have used the fact that $ \norm[2]{\overline{\bq}'(\tau)}^2 +{\text {Re}}\left\{ \overline {\bq}''(\tau)^H \overline{\bq}(\tau) \right\} \leq -0.3756 M^2$ from Appendix \ref{sec:suppmaterial}; see also  \cite{yang2014exact} for a similar argument. The fifth line follows because $|K_M''(0)| < \frac{4\pi^2 M^2}{3}$ and the last line follows by choosing $\varepsilon$ small enough, say $\varepsilon = 10^{-5}$. This completes the proof of Lemma \ref{lem:nearregionproof}.
\end{proof}
Combining Lemmas \ref{lem:farregionproof} and \ref{lem:nearregionproof} immediately shows that $\norm[2]{\bq(\tau)}<1$ everywhere, $\tau \in [0,1)\setminus\mathbb{D}$.

\section{Conclusions}
\label{sec:conclusions}
In this paper, we developed a new model for non-stationary blind super-resolution. In our model, the problem is naturally non-convex. Using the lifting trick, we formulated the problem as a convex program under a subspace assumption for the unknown waveforms. A sample complexity bound that is proportional to the number of degrees of freedom in the problem was derived for exact recovery under the condition that the locations of the point sources are sufficiently separated. Numerical simulations were provided to validate our proposed approach. Future directions include extending our model and provable recovery guarantee to the noisy data case, testing our proposed approach for real data applications such as blind super-resolution for single-molecule imaging, and relaxing the subspace assumption to a sparse dictionary model.

\section{Acknowledgements}

The authors would like to thank Yuejie Chi for helpful discussions on the application of the proposed model to blind multi-path channel identification in communication systems.

\appendix

\section{Proof of Proposition~\ref{optimalitycond}}
\label{sec:proofofmainthm}
The proof strategy follows quite straightforwardly from that in Proposition II.4 of  \cite{tang2013compressed}. First of all, any dual vector $\blambda$ satisfying (\ref{cond1}) and (\ref{cond2}) in Proposition \ref{optimalitycond} is dual feasible. To see this, note that
\begin{equation*}
\begin{aligned}
\|\cB^{*}\left(\blambda\right)\|_{\cA}^{*} & = \sup_{\norm[\cA]{\bX}\leq 1} \langle \cB^{*}\left(\blambda\right), \bX \rangle_{\real} \\
& = \sup_{\tau\in[0,1), \norm{\bh} = 1} \langle \cB^{*}\left(\blambda\right), \bh \ba(\tau)^H \rangle_{\real} \\
&  = \sup_{\tau\in[0,1), \norm{\bh} = 1} \left\langle \sum_n \blambda(n) \bb_n \be_n^H, \bh \ba(\tau)^H \right\rangle_{\real} \\
& = \sup_{\tau\in[0,1), \norm{\bh} = 1} {\text{Re}} (\bh^H \bq(\tau)) \\
& \leq \sup_{\tau\in[0,1)} \norm[2]{\bq(\tau)} \\
& \leq 1.
\end{aligned}
\end{equation*}
Here, the second equality comes from the fact that the atoms $\left\{\bh \ba(\tau)^H \right\}$ comprise all of the extremal points of the atomic unit ball $\left\{\bX:\norm[\cA]{\bX} \leq 1\right\}$.
Furthermore, for any $\blambda$ that satisfies (\ref{cond1}) in Proposition \ref{optimalitycond}, we have
\begin{align*}
\langle \blambda, \by \rangle_{\real} & =  \langle \cB^{*}(\blambda), \bX_o\rangle_{\mathbb{R}} \\
&  = \left\langle \cB^{*}(\blambda), \sum_{j=1}^J c_j \bh_j \ba (\tau_j)^H \right\rangle_{\real}\\
& = \sum_{j=1}^J {\text {Re}} \left(c_j^{*}~{\trace}\left( \ba(\tau_j) \bh_j^H\cB^{*}(\blambda) \right)\right)\\
& = \sum_{j = 1}^J {\text { Re}}\left( c_j^{*} \bh_j^H \bq(\tau_j) \right) \notag \displaybreak\\
& = \sum_{j= 1}^J {\text {Re}} \left( c_j^{*} \sign(c_j) \right)\\
& = \sum_{j = 1}^J |c_j|  \geq \| \bX_o\|_{\mathcal{A}},
\end{align*}
where the fifth line follows from (\ref{cond1}) in Proposition \ref{optimalitycond}, and the last line follows from the definition of the atomic norm. On the other hand, by H{\"o}lder's inequality, we have
\begin{equation*}
\begin{aligned}
\langle \blambda, \by \rangle_{\real} & = \langle \cB^{*} (\blambda), \bX_o\rangle_{\real} \\
& \leq \| \cB^{*}(\blambda)\|_{\cA}^{*} \norm[\cA]{\bX_o} \\
& \leq \norm[\cA]{\bX_o}.
\end{aligned}
\end{equation*}
This implies that any $\blambda$ satisfying (\ref{cond1}) and (\ref{cond2}) also satisfies $\langle \blambda, \by \rangle_{\real} = \| \bX_o\|_{\cA}$. Therefore, we have shown that the primal and dual problems have zero duality gap achieved by the primal feasible solution $\bX_o$ and dual feasible solution $\blambda$, which means that $\blambda$ is dual optimal and $\bX_o$ is primal optimal.

Finally, as we show below, condition (\ref{cond2}) ensures that $\bX_o$ is the unique optimal solution. To see this, suppose that there exists another optimal solution $\widetilde{\bX} = \sum_{j} \widetilde{c}_j \widetilde{\bh}_j \ba(\widetilde{\tau}_j)^H$. Then we can see that
\begin{equation*}
\begin{aligned}
\left\langle \blambda, \by \right\rangle_{\mathbb{R}}  & =  \left\langle  \mathcal{B}^{*}(\blambda), \widetilde{\bX} \right\rangle_{\mathbb{R}} \\
 & = \left\langle \mathcal{B}^{*}(\blambda), \sum_j \widetilde{c}_j \widetilde{\bh}_j \ba(\widetilde{\tau}_j)^H\right\rangle_{\mathbb{R}}\\
& = \sum_{\tau_j \in \mathbb{D}} {\text {Re}}\left( \widetilde{c}_j^* \widetilde{\bh}_j^H  \bq(\widetilde{\tau}_j)\right) +  \sum_{\tau_{\ell} \notin \mathbb{D}} {\text{Re}}\left( \widetilde{c}_{\ell}^{*}  \widetilde{\bh}_{\ell}^H  \bq(\widetilde{\tau}_{\ell}) \right)\\
& < \sum_{\tau_j \in \mathbb{D}} |\widetilde{c}_j| + \sum_{\tau_{\ell} \notin \mathbb{D}} |\widetilde{c}_{\ell}| \\
&  = \|\widetilde{\bX}\|_{\mathcal{A}}.
\end{aligned}
\end{equation*}
To show that $\{ c_j\bh_j \}$ are also unique, we can form the following linear system of equations:
\begin{equation*}
\begin{aligned}
& \begin{bmatrix}
 \ba(\tau_1)^H  \be_{-2M} \bb_{-2M}^H & \cdots &  \ba(\tau_J)^H  \be_{-2M} \bb_{-2M}^H \\
\vdots  & \ddots &  \vdots \\
\ba(\tau_1)^H  \be_{2M} \bb_{2M}^H & \cdots &  \ba(\tau_J)^H  \be_{2M} \bb_{2M}^H \\
\end{bmatrix} \begin{bmatrix} c_1 \bh_1 \\ \vdots \\ c_J \bh_J \end{bmatrix} =   \begin{bmatrix} \by(-2M) \\ \vdots \\ \by(2M)  \end{bmatrix}.
\end{aligned}
\end{equation*}
The linearly independent condition in Proposition \ref{optimalitycond} ensures that $\{c_j \bh_j, j = 1, \cdots, J \}$ are unique. Thus, $\bX_o$ is the unique optimal solution of the atomic norm minimization (\ref{eq:mainprog}) if conditions (\ref{cond1}), (\ref{cond2}), and the linearly independent condition are satisfied.

\section{Proof of Lemma~\ref{le:lemma1}}
\label{sec:prooflemma1}
We need the following supporting lemmas for proving Lemma~\ref{le:lemma1}.
\begin{lemma}\label{lem:invertiblekroneck}
For arbitrary two matrices $\bA$ and $\bB$, the non-zero singular values of their Kronecker product $\bA \otimes \bB$ are $\sigma_{i}(\bA) \sigma_{j}(\bB)$, where $\sigma_i(\bA)$ and $\sigma_j(\bB)$ are the non-zero singular values of $\bA$ and $\bB$, respectively. In particular, we have
\begin{equation*}
\norm[]{\bA \otimes \bB} = \norm[]{\bA} \norm[]{\bB}.
\end{equation*}
\end{lemma}
\begin{lemma}\label{lem:fejerkernelinvertible}
\cite{candes2014towards} Suppose $\Delta_{\tau} \geq \frac{1}{M}$. Then $\bD'$ is invertible and
\begin{equation*}
\norm[]{\bI_{2J}-\bD'}\leq 0.3623,
\end{equation*}
\begin{equation*}
\norm[]{\bD'}\leq 1.3623,
\end{equation*}
\begin{equation*}
\norm[]{\bD'^{-1}}\leq 1.568.
\end{equation*}
\end{lemma}
According to Lemma \ref{lem:invertiblekroneck}, we have
\begin{equation*}
\begin{aligned}
\norm[]{\expval \bD} &  = \norm[]{\bD' \otimes \bI_K}\\
& = \norm[]{\bD'}\\
& \leq 1.3623
\end{aligned}
\end{equation*}
and
\begin{equation*}
\begin{aligned}
\norm[]{\bI_{2JK}-\expval \bD} &  =\norm[]{(\bI_{2J}-\bD')\otimes \bI_K}\\
& = \norm[]{\bI_{2J}-\bD'} \\
& \leq 0.3623.
\end{aligned}
\end{equation*}
As a result, we have
\begin{equation*}
\norm[]{\left(\expval \bD\right)^{-1}} \leq 1.568.
\end{equation*}
\section{Proof of Lemma~\ref{lem:matroperatornormbound}}
\label{operatornormboundonV}
We use the matrix Bernstein inequality for proving Lemma \ref{lem:matroperatornormbound}.
\begin{lemma}
\label{lem:joeltropp}
\cite{tropp2012user} (Matrix Bernstein: Rectangular Case) Consider a finite sequence $\left\{ \bX_k\right\}$ of independent, random matrices with dimension $d_1\times d_2$. Assume that each random matrix satisfies
\begin{equation*}
\mathbb{E} \bX_k = 0~~~and~~~\| \bX_k\| \leq R~~~almost~surely.
\end{equation*}
Define
\begin{equation*}
\sigma^2 := \max \left\{ \| \sum_k \mathbb{E} (\bX_k \bX_k^{H})\|, \|\sum_k \mathbb{E} (\bX_k^{H} \bX_k)\|   \right\}.
\end{equation*}
Then, for all $t\geq 0$,
\begin{equation*}
\begin{aligned}
 \mathbb{P} \left\{ \| \sum_{k} \bX_k\|  \geq t\right\}
  &\leq  (d_1 + d_2) \cdot \exp\left(\frac{-t^2/2}{\sigma^2 + Rt/3}\right)\\
&\leq
{\begin{cases}
(d_1 + d_2) \exp \left(\frac{-3t^2}{8\sigma^2} \right), & t \leq \sigma^2 /R \\
(d_1 + d_2) \exp \left(\frac{-3t}{8R} \right), & t \geq \sigma^2/R.
\end{cases}}
\end{aligned}
\end{equation*}
\end{lemma}
First of all, we can write
\begin{equation*}
\begin{aligned}
 \bV_{\ell}(\tau) - \mathbb{E}\bV_{\ell}(\tau) 
 & =   \frac{1}{M} \sum_{n = -2M}^{2M} g_M(n) \left( \frac{-i2\pi n}{\sqrt{|K_M^{''}(0)|}} \right)^{\ell} e^{-i 2\pi n \tau } \bE(n) \otimes \left(\bb_n \bb_n^H - \bI_{K}\right)\\
 & =  \sum_{n = -2M}^{2M} \bY_n^{\ell},
\end{aligned}
\end{equation*}
where we have defined
\begin{equation*}
\begin{aligned}
\bY_{n}^{\ell} & =  \frac{1}{M} g_M(n) \left( \frac{-i2\pi n}{\sqrt{|K_M^{''}(0)|}} \right)^{\ell} \times e^{-i 2\pi n \tau} \bE(n) \otimes \left(\bb_n \bb_n^H - \bI_{K}\right).
\end{aligned}
\end{equation*}
It is easy to see that $\left\{ \bY_{n}^{\ell}\right\}$ are independent random matrices with zero mean due to the isotropy properties of $\bb_n$. Thus, we apply the matrix Bernstein inequality for bounding $\norm[]{\bV_{\ell}(\tau) -\expval \bV_{\ell}(\tau)}$.
Before establishing this, we need to compute the quantities $R$ and $\sigma^2$ in the matrix Bernstein inequality. The following elementary bound  \cite{candes2014towards,tang2013compressed} will be useful at this moment:
\begin{equation*}
\| g_M(n)\|_{\infty} \leq 1,
\end{equation*}
\begin{equation*}
\left| \frac{2\pi n}{\sqrt{|K_M^{''}(0)|}}\right| \leq 4,~~~~~~{\text{when}}~M\geq 2,
\end{equation*}
and
\begin{equation*}
\|\bE(n)\|_2^2 \leq 14 J,~~~~~~{\text{when}}~M\geq 4.
\end{equation*}
Thus, we have
\begin{equation*}
\begin{aligned}
\norm[]{\bY_n^{\ell}} & = \norm[]{ \frac{1}{M} g_M(n) \left( \frac{-i2\pi n}{\sqrt{|K_M^{''}(0)|}} \right)^{\ell} \vphantom{\frac{1}{M} g_M(n) \left( \frac{-i2\pi n}{\sqrt{|K_M^{''}(0)|}} \right)^{\ell} } e^{-i 2\pi n \tau } \bE(n) \otimes \left(\bb_n \bb_n^H - \bI_{K}\right)} \\
& \leq \frac{1}{M} 4^{\ell} \norm[2]{\bE(n)} \norm[]{\bb_n \bb_n^H - \bI_{K}} \\
& \leq \frac{1}{M} 4^{\ell} \sqrt{14J} \max \left\{ \mu K,1 \right\} \\
& \leq \frac{4^{\ell +1} \sqrt{J} \mu K}{M} \\
& =: R,
\end{aligned}
\end{equation*}
where the second line uses the fact that $\|g_M(n)\|_{\infty} \leq 1$, $| \frac{-i 2\pi n}{\sqrt{|K_M^{''}(0)|}}| \leq 4$ and $\norm[]{\bA \otimes \bB} = \norm[]{\bA} \norm[]{\bB}$ for arbitrary two matrices $\bA$ and $\bB$. The third line follows from the fact that $\norm[2]{\bE(n)}^2 \leq 14J$ and $\norm[]{\bA -\bB} \leq \max\left\{\norm[]{\bA}, \norm[]{\bB} \right\}$ for two positive semidefinite matrices $\bA$ and $\bB$. The fourth line uses the assumption that $\mu K\geq 1$.

Then, we compute the variance term
\begin{align*}
 \norm[]{\sum_n \expval {\bY_n^{\ell}}^H \bY_n^{\ell}} 
 &= \frac{1}{M^2}\norm[]{\sum_{n=-2M}^{2M} \expval |g_M(n)|^2 \left( \left| \frac{-i2\pi n}{\sqrt{|K_M^{''}(0)|}} \right|\right)^{2 \ell} \right. \\ \left.
 \right.& \left.~~~~~\vphantom{\sum_{n=-2M}^{2M} \expval |g_M(n)|^2 \left( \left| \frac{-i2\pi n}{\sqrt{|K_M^{''}(0)|}} \right|\right)^{2 \ell}}\times \left (\bE_n^H \otimes \left(\bb_n \bb_n^H -\bI_K\right)\right)\left( \bE_n \otimes (\bb_n\bb_n^H - \bI_K)\right)} \\
 &\leq  \frac{1}{M^2} 4^{2\ell} \norm[]{\sum_{n =-2M}^{2M} \norm[2]{\bE_n}^2  \expval (\bb_n\bb_n^H -\bI_K)^2} \notag \displaybreak\\
 &\leq  \frac{4^{2\ell} \mu K}{M^2} \norm[]{\sum_{n = -2M}^{2M} \norm[2]{\bE_n}^2  \bI_K} \\
& = \frac{4^{2 \ell} \mu K}{M^2} \sum_{n=-2M}^{2M} \norm[2]{\bE_n}^2 \\
 &\leq  \frac{80\cdot4^{2 \ell} \mu J K}{M} \\
 &=:  \sigma^2,
\end{align*}
where the second line follows from the fact that $\|g_M(n)\|_{\infty}\leq 1$ and $| \frac{-i 2\pi n}{\sqrt{|K_M^{''}(0)|}}| \leq 4$, the third line uses the fact that $\norm[2]{\bb_n}^2 \bb_n \bb_n^H \preceq \mu K \bb_n \bb_n^H, \mu K\geq 1$ and that $\expval \bb_n \bb_n^H = \bI_K$ due to the incoherence property (\ref{eq:incoherenceproperty}) and the isotropy property (\ref{eq:isotropyproperty}), and the last inequality uses the fact that $\norm[2]{\bE(n)}^2 \leq 14J$ when $M\geq 4$.

Applying Lemma \ref{lem:joeltropp}, we can see that for a fixed $\ell$,
\begin{equation*}
\begin{aligned}
\mathbb{P} \left\{ \| \sum_{n} \bY_n^{\ell}\|  \geq \varepsilon_2 \right\} & \leq (2JK + K) \cdot \exp \left(\frac{-3\varepsilon_2^2}{8\sigma^2} \right). \\
\end{aligned}
\end{equation*}
In order to make this failure probability less than $\delta_2$, we require
\begin{equation*}
\log \left((2JK +K) \cdot \exp \left(\frac{-3\varepsilon_2^2}{8\sigma^2} \right)\right) \leq \log \delta_2,
\end{equation*}
which leads to the following bound on $M$,
\begin{equation*}
\begin{aligned}
M & \geq \frac{640 \cdot 4^{2\ell} \mu JK} {3\varepsilon_2^2}\log\left(\frac{2JK +K}{\delta_2}\right). \\
\end{aligned}
\end{equation*}
Applying a union bound for $\ell = 0,1,2,3$, we obtain that
\begin{equation*}
\begin{aligned}
\mathbb{P} \left\{ \| \sum_{n} \bY_n^{\ell}\|  \geq \varepsilon_2, \ell = 0, 1, 2, 3 \right\} & \leq 4 \delta_2,\\
\end{aligned}
\end{equation*}
provided that $M  \geq \frac{640 \cdot 4^{2\ell} \mu JK} {3\varepsilon_2^2}\log\left(\frac{2JK +K}{\delta_2}\right)$. This completes the proof.

\section{Proof of Lemma~\ref{lem:unionfailprob}}
\label{proofofLemma5_6}
In Lemma \ref{lem:matroperatornormbound}, we showed that for $\ell = 0, 1, 2, 3$,
$\norm[]{\bV_{\ell}(\tau) - \expval \bV_{\ell}(\tau)} \leq \varepsilon_2$ with probability at least $1-4\delta_2$ provided $M \geq \frac{640\cdot 4^{2\ell }\mu JK}{3\varepsilon_2^2} \log \left(\frac{2JK+K}{\delta_2}\right)$. Conditioned on the events $\cE_{1, \varepsilon_1}$ with $\varepsilon_1 \in (0,\frac{1}{4}]$ and 
\begin{equation*}
\bigcap_{\tau_d \in \Omega_{\Grid}} \left\{ \norm[]{\bV_{\ell}(\tau_d) - \expval \bV_{\ell}(\tau_d)} \leq \varepsilon_2, \ell = 0,1,2,3\right\}
\end{equation*}
we have
\begin{equation*}
\begin{aligned}
\norm[]{\left(\bV_{\ell}(\tau_d) - \expval \bV_{\ell}(\tau_d)\right)^H \bL} 
& \leq  \norm[]{\bV_{\ell}(\tau_d) - \expval \bV_{\ell}(\tau_d)} \norm[]{\bL} \\
 &\leq \norm[]{\bV_{\ell}(\tau_d) - \expval \bV_{\ell}(\tau_d)} \norm[]{\bD^{-1}} \\
 &\leq  \varepsilon_2 2 \norm[]{(\expval\bD)^{-1}} \\
& \leq  4\varepsilon_2,
\end{aligned}
\end{equation*}
where the third line uses the fact that $\bL$ is a submatrix of $\bD^{-1}$, and the fourth and fifth lines follow from Lemmas \ref{lem:inverseopnorm} and \ref{le:lemma1}, respectively.

Applying the union bound leads to
\begin{equation*}
\begin{aligned}
\mathbb{P} \left\{ \sup_{\tau_d \in \Omega_{\Grid}} \norm[]{\left(\bV_{\ell}(\tau_d) - \expval \bV_{\ell}(\tau_d)\right)^H \bL} \geq 4 \varepsilon_2,\vphantom{ \sup_{\tau_d \in \Omega_{\Grid}} \norm[]{\left(\bV_{\ell}(\tau_d) - \expval \bV_{\ell}(\tau_d)\right)^H \bL} }\ell = 0, 1, 2, 3 \right\} \leq |\Omega_{\Grid}|4 \delta_2 + \mathbb{P}\left(\cE_{1, \varepsilon_1}^c\right).
\end{aligned}
\end{equation*}

\section{Proof of Lemma~\ref{lem:vectorbernstein}}
\label{proofofLemma5_8}
We need the following lemma in the proof of Lemma \ref{lem:vectorbernstein}.
\begin{lemma}
\label{lem:expvalunitsphere}
Assume that $\bh_j \in \complex^K$ are i.i.d.\ random samples on the complex unit sphere $\mathbb{CS}^{K-1}$. Then we have $\expval \bh_j \bh_j^H  = \frac{1}{K} \bI_K$.
\end{lemma}
\begin{proof} Denote $\bSigma = \expval \bh_j \bh_j^H$. By unitary invariance, we have $\expval \bU \bh_j \left( \bU \bh_j \right)^H = \bSigma$, which implies that $\bU \bSigma \bU^H = \bSigma$ for any unitary matrix $\bU$. This indicates that $\bSigma$ is diagonal. Furthermore, if $\bU$ is a permutation matrix, $\bU \bSigma\bU^H$ permutes the diagonal entries of $\bSigma$. This shows that the diagonal entries of $\bSigma$ have equal values. Lastly, $\trace(\bSigma) = \expval \trace (\bh_j \bh_j^H) = 1$. Thus, we have $\bSigma = \frac{1}{K} \bI_K$.
\end{proof}
For any $\tau_d \in \Omega_{\Grid}$, define
\begin{equation*}
\begin{aligned}
\bQ & : = (\bV_{\ell}(\tau_d) - \mathbb{E} \bV_{\ell}(\tau_d) )^H \bL \\
&  = \begin{bmatrix} \bQ_1 & \bQ_2 & \dots & \bQ_J \end{bmatrix},
\end{aligned}
\end{equation*}
where each block $\bQ_j$ is a $K\times K$ matrix. Also, define the event
\begin{equation*}
\begin{aligned}
\cE_{3}: = \left\{ \sup_{\tau_d \in \Omega_{\Grid}}  \norm[]{\left(\bV_{\ell}(\tau_d) - \expval \bV_{\ell}(\tau_d)\right)^H \bL} \leq 4 \varepsilon_2,\vphantom{ \sup_{\tau_d \in \Omega_{\Grid}} \norm[]{\left(\bV_{\ell}(\tau_d) - \expval \bV_{\ell}(\tau_d)\right)^H \bL} }
\ell = 0, 1, 2, 3 \right\}.
\end{aligned}
\end{equation*}
We can write
\begin{equation*}
\begin{aligned}
\bI_1^{\ell}(\tau_d) & = (\bV_{\ell}(\tau_d) - \mathbb{E} \bV_{\ell}(\tau_d) )^H \bL \bh \\
& = \sum_j \bQ_j \sign(c_j )\bh_j \\
& =: \sum_j \bZ_j.
\end{aligned}
\end{equation*}
Note that $\expval{\bZ_j} = \bzero_{K\times 1}$ due to the randomness assumption of $\bh_j$. Before applying the matrix Bernstein inequality for bounding $\norm[]{\sum_j \bZ_j}$, we need to upper bound the operator norm $\norm[]{\bZ_j}$ and compute the variance term appearing in the expression of matrix Bernstein inequality.  Conditioned on $\cE_3$,
\begin{equation*}
\begin{aligned}
\norm[2]{\bZ_j} & = \norm[2]{\bQ_j \sign(c_j) \bh_j} \\
& \leq \norm[]{\bQ_j}\\
 & \leq \norm[]{\bQ}  \\
& \leq 4 \varepsilon_2 \\
& =: R
\end{aligned}
\end{equation*}
where the third line uses the fact that $\bQ_j$ is a submatrix of $\bQ$.

Next, conditioned on the event $\cE_{3}$ (note that event $\cE_3$ includes event $\cE_{1, \varepsilon_1 \in (0,\frac{1}{4}]}$ and $\cE_{2,\varepsilon_2}$), we bound the variance term:
\begin{equation*}
\begin{aligned}
 \norm[]{\sum_j \expval \bZ_j^H \bZ_j } & = \norm[]{\sum_j \expval (\sign(c_j)\bh_j)^H \bQ_j^H \bQ_j \sign(c_j) \bh_j} \\
&  = \sum_j \expval \trace\left( \bQ_j^H \bQ_j \bh_j \bh_j^H\right) \\
& = \sum_j \trace\left( \bQ_j^H \bQ_j \expval \left[\bh_j \bh_j^H \right] \right) \\
& = \sum_j \trace\left( \bQ_j^H \bQ_j \frac{1}{K} \bI_K \right) \\
&  = \frac{1}{K} \norm[F]{\bQ}^2,
\end{aligned}
\end{equation*}
where the third line follows by exchanging the order of the trace operation and expectation, the fourth line uses Lemma \ref{lem:expvalunitsphere}. Furthermore, we can bound $\frac{1}{K} \norm[F]{\bQ}^2 $ as follows:
\begin{equation*}
\begin{aligned}
\frac{1}{K} \norm[F]{\bQ}^2 &  \leq \frac{1}{K} \norm[]{\bL^H}^2 \norm[F]{\bV_{\ell}(\tau_d) - \expval \bV_{\ell}(\tau_d)}^2 \\& \leq \frac{4 \cdot 1.568^2}{K}  K \varepsilon_2^2 \\
& \leq 12 \varepsilon_2^2 \\
& =:  \sigma^2,
\end{aligned}
\end{equation*}
where the first line follows from the fact that $\norm[F]{\bA \bB}^2 \leq \norm[]{\bA}^2\norm[F]{\bB}^2$ for arbitrary two matrices $\bA$ and $\bB$, the second line follows from the fact that $\bL$ is a submatrix of $ \bD^{-1}$ and Lemmas \ref{lem:froboperatornormineq}, \ref{lem:inverseopnorm} and \ref{le:lemma1}.

Applying the matrix Bernstein inequality and the union bound, we get
\begin{equation*}
\begin{aligned}
\mathbb{P} \left\{ \sup_{\tau_d \in \Omega_{\Grid}}\norm[2]{\bI_1^{\ell}(\tau_d)} \geq \varepsilon_4 \Bigg| \cE_{3} \right\} 
& \leq |\Omega_{\Grid}|
\mathbb{P} \left\{ \| \sum_{j} \bZ_j\|  \geq \varepsilon_4 \Bigg|  \cE_3 \right\} \\
& \leq |\Omega_{\Grid}| (K + 1) \cdot \exp\left(\frac{-\varepsilon_4^2/2}{\sigma^2 + R\varepsilon_4 /3}\right) \\
& \leq
{\begin{cases}
|\Omega_{\Grid}|(K + 1) \exp \left(\frac{-3\varepsilon_4^2}{8\sigma^2} \right), & \varepsilon_4 \leq \sigma^2 /R \\
|\Omega_{\Grid}|(K + 1) \exp \left(\frac{-3\varepsilon_4}{8R} \right), & \varepsilon_4 \geq \sigma^2/R.
\end{cases}}
\end{aligned}
\end{equation*}
Taking $\varepsilon_2^2 = \frac{640 \cdot 4^{2\ell} \mu JK} {3M}\log\left(\frac{2JK +K}{ \delta_2}\right)$ and applying Lemma \ref{lem:unionfailprob} yield
\begin{equation*}
\begin{aligned}
& \mathbb{P} \left\{ \sup_{\tau_d \in \Omega_{\Grid}}\norm[2]{\bI_1^{\ell}(\tau_d)}\geq \varepsilon_4 \right\} \\
&  \leq {\begin{cases}
\begin{aligned}|\Omega_{\Grid}|(K + 1) \exp \left(\frac{-3\varepsilon_4^2}{8\sigma^2} \right)   +|\Omega_{\Grid}|4 \delta_2
+ \mathbb{P}\left(\cE_{1, \varepsilon_1}^c\right),
\end{aligned} & \varepsilon_4 \leq \sigma^2 /R \\
\begin{aligned}
|\Omega_{\Grid}|(K + 1) \exp \left(\frac{-3\varepsilon_4}{8R} \right) 
+|\Omega_{\Grid}|4 \delta_2 + \mathbb{P}\left(\cE_{1, \varepsilon_1}^c\right),
\end{aligned}& \varepsilon_4 \geq \sigma^2/R.
\end{cases}}
\end{aligned}
\end{equation*}
According to Lemmas \ref{lem:matroperatornormbound} and \ref{lem:unionfailprob}, for the second term $|\Omega_{\Grid}| 4\delta_2 \leq \delta$, it is sufficient to have
\begin{equation*}
M \geq \frac{640 \cdot 4^{2\ell }\mu JK}{3\varepsilon_2^2} \log \left(\frac{4 |\Omega_{\Grid}|(2JK+K)}{\delta}\right).
\end{equation*}
To make the failure probability less than or equal to $\delta$ for the first term, we choose
\begin{equation*}
\begin{cases}
96\varepsilon_2^2 = \frac{3\varepsilon_4^2} {\log \left(\frac{|\Omega_{\Grid}|(K+1)}{\delta}\right)}, & \varepsilon_4 \leq \sigma^2 /R,\\
32\varepsilon_2 = \frac{3\varepsilon_4} {\log \left(\frac{|\Omega_{\Grid}|(K+1)}{\delta}\right)}, & \varepsilon_4 \geq \sigma^2 /R.
\end{cases}
\end{equation*}
Equivalently, when $\varepsilon_4 \leq \sigma^2 /R$, one has
\begin{equation*}
\begin{aligned}
M & \geq \frac{640 \cdot 4^{2\ell }\mu JK}{3\varepsilon_2^2} \log \left(\frac{4 |\Omega_{\Grid}|(2JK+K)}{\delta}\right) \\
& = \frac{640 \cdot 96\cdot 4^{2\ell} \mu JK} {9 \varepsilon_4^2}\log\left(\frac{4|\Omega_{\Grid}|(2JK +K)}{ \delta}\right) \log \left(\frac{|\Omega_{\Grid}|(K+1)}{\delta}\right).
\end{aligned}
\end{equation*}
When $\varepsilon_4 \geq \sigma^2 /R$, one has
\begin{equation*}
\begin{aligned}
M & \geq \frac{640 \cdot 4^{2\ell }\mu JK}{3\varepsilon_2^2} \log \left(\frac{4 |\Omega_{\Grid}|(2JK+K)}{\delta}\right) \\
& = \frac{32^2 \cdot 640 \cdot 4^{2\ell} \mu JK} {27 \varepsilon_4^2}\log\left(\frac{4|\Omega_{\Grid}|(2JK +K)}{ \delta}\right) \log^2 \left(\frac{|\Omega_{\Grid}|(K+1)}{\delta}\right).
\end{aligned}
\end{equation*}
Finally, according to Lemma \ref{lem:matrixbernstein}, for the third term $\mathbb{P}\left(\cE_{1,\varepsilon_1}^c\right) \leq \delta$, we have
\begin{equation*}
M \geq \frac{80 \mu JK}{\varepsilon_1^2} \log\left( \frac{4JK}{\delta}\right).
\end{equation*}
Setting $\varepsilon_1 = \frac{1}{4}$, absorbing all of the constants into one and applying the union bound for $\ell = 0,1,2, 3$, we can see that
\begin{equation*}
\begin{aligned}
\mathbb{P} \left\{ \sup_{\tau_d \in \Omega_{\Grid}}\norm[2]{\bI_1^{\ell}(\tau_d)}\geq \varepsilon_4, \ell = 0, 1, 2, 3\right\} \leq 12 \delta\\
\end{aligned}
\end{equation*}
provided
\begin{equation*}
\begin{aligned}
M \geq C\mu JK \max\left\{\frac{1}{\varepsilon_4^2} \log\left(\frac{|\Omega_{\Grid}|JK}{\delta}\right) \log^2 \left(\frac{|\Omega_{\Grid}|K}{\delta}\right),  \log\left(\frac{JK}{\delta}\right)\right\}
\end{aligned}
\end{equation*}
for some constant $C$.

\section{Proof of Lemma~\ref{lem:frobeniusnorm2}}
\label{appendix:lemma5.9}
Observe that
\begin{align*}
\norm[F]{\expval \bV_{\ell}(\tau)}^2 
&  = \norm[F]{ \frac{1}{\sqrt{|K''_M(0)|}^{\ell}} \begin{bmatrix}
K_M^{\ell} (\tau-\tau_1)^{*} \\
\vdots \\
K_M^{\ell} (\tau-\tau_J)^{*} \\
\frac{1}{\sqrt{|K''_M(0)|}}K_M^{\ell +1}(\tau-\tau_1)^{*}\\
\vdots \\
\frac{1}{\sqrt{|K''_M(0)|}}K_M^{\ell +1}(\tau-\tau_J)^{*}\\
\end{bmatrix} \otimes \bI_{K}}^2\notag \displaybreak\\
&  = K  \norm[2]{ \frac{1}{\sqrt{|K''_M(0)|}^{\ell}} \begin{bmatrix}
K_M^{\ell} (\tau-\tau_1)^{*} \\
\vdots \\
K_M^{\ell} (\tau-\tau_J)^{*} \\
\frac{1}{\sqrt{|K''_M(0)|}}K_M^{\ell +1}(\tau-\tau_1)^{*}\\
\vdots \\
\frac{1}{\sqrt{|K''_M(0)|}}K_M^{\ell +1}(\tau-\tau_J)^{*}\\
\end{bmatrix} }^2 \\
& \leq C K
\end{align*}
for some numerical constant $C$, where the inequality above follows from proof of Lemma IV.9 in \cite{tang2013compressed}. The key to being able to obtain such a bound of order $O(K)$ is because $\left\{\tau_j\right\}$ are well separated, implying that the sequence $\left\{K_M^{\ell}(\tau-\tau_j)\right\}$ decreases rapidly if properly ordered. Then, conditioned on the event $\cE_{1,\varepsilon_1}$ with $\varepsilon_1 \in (0,\frac{1}{4}]$, we have
\begin{equation*}
\begin{aligned}
\norm[F]{\left(\bL -\bL'\otimes \bI_K \right)^H\expval \bV_{\ell}(\tau)}^2 & \leq \norm[]{\left(\bL -\bL'\otimes \bI_K \right)^H}^2 \norm[F]{\expval \bV_{\ell}(\tau)}^2 \\
& \leq (2\cdot 1.568^2 \varepsilon_1)^2  C K \\
& =: C K\varepsilon_1^2
\end{aligned}
\end{equation*}
for some redefined numerical constant $C$. The first inequality above uses the fact that $\norm[F]{\bA \bB}^2 \leq \norm[]{\bA}^2 \norm[F]{\bB}^2$ for two arbitrary matrices $\bA$ and $\bB$, and the second inequality above follows from the fact that $\bL - \bL'\otimes \bI_K$ is a submatrix of $\bD^{-1} - \expval \bD^{-1}$ and from Lemmas \ref{lem:inverseopnorm} and \ref{le:lemma1}.

\section{Proof of Lemma~\ref{lem:boundI2}}
\label{appendix:lemma5.10}
To begin with,
for any $\tau_d \in \Omega_{\Grid}$, define
\begin{equation*}
\begin{aligned}
\widetilde{\bQ} & : = \left[\expval \bV_{\ell}(\tau_d)\right]^H \left(\bL-\bL' \otimes \bI_K\right) \\
&  = \begin{bmatrix} \widetilde{\bQ}_1 & \widetilde{\bQ}_2 & \dots & \widetilde{\bQ}_J \end{bmatrix},
\end{aligned}
\end{equation*}
where each block $\widetilde{\bQ}_j$ is a $K\times K$ matrix. Then, we have
\begin{equation*}
\begin{aligned}
\bI_2^{\ell}(\tau_d) & = \left[ \expval \bV_{\ell}(\tau_d)\right]^H \left(\bL-\bL' \otimes \bI_K\right) \bh\\
& = \sum_j \widetilde{\bQ}_j \sign(c_j) \bh_j \\
& =: \sum_j \widetilde{\bZ}_j.
\end{aligned}
\end{equation*}
Again, we bound $\norm[2]{\bI_2^{\ell}(\tau_d)}$ using the matrix Bernstein inequality. First of all, we have $\expval \widetilde{\bZ}_j = \bzero_{K\times 1}$ due to the randomness assumption of $\bh_j$. Conditioned on $\cE_{1, \varepsilon_1}$ with $\varepsilon_1 \in (0,\frac{1}{4}]$,
\begin{align*}
\norm[]{\widetilde{\bZ}_j}
 & = \norm[]{\widetilde{\bQ}_j \sign(c_j) \bh_j} \\
& \leq \|\widetilde{\bQ}\| \\
& \leq  \norm[]{\bL -\bL'\otimes \bI_K } \norm[]{\expval \bV_{\ell}(\tau_d)} \notag \displaybreak\\
&  = \frac{2\cdot 1.568^2 \varepsilon_1}{\sqrt{|K''_M(0)|}^{\ell}}   \norm[2]{ \begin{bmatrix}
K_M^{\ell} (\tau-\tau_1)^{*} \\
\vdots \\
K_M^{\ell} (\tau-\tau_J)^{*} \\
\frac{1}{\sqrt{|K''_M(0)|}}K_M^{\ell +1}(\tau-\tau_1)^{*}\\
\vdots \\
\frac{1}{\sqrt{|K''_M(0)|}}K_M^{\ell +1}(\tau-\tau_J)^{*}\\
\end{bmatrix} }  \\
& \leq C \varepsilon_1 \\
&  =: R.
\end{align*}
for some small universal constant $C$.  For the variance term, we have
\begin{equation*}
\begin{aligned}
\norm[]{\sum_j \expval \widetilde{\bZ}_j^H \widetilde{\bZ}_j} & = \norm[]{\sum_j \expval (\sign(c_j)\bh_j)^H \widetilde{\bQ}_j^H \widetilde{\bQ}_j \sign(c_j) \bh_j} \\
&  = \sum_j \expval \trace\left( \widetilde{\bQ}_j^H \widetilde{\bQ}_j \bh_j \bh_j^H\right) \\
& = \sum_j \trace\left(\widetilde{\bQ}_j^H \widetilde{\bQ}_j \expval\left[ \bh_j \bh_j^H  \right]\right) \\
& = \sum_j \trace\left( \widetilde{\bQ}_j^H \widetilde{\bQ}_j \frac{1}{K} \bI_K \right) \\
 & \leq \frac{1}{K} C K\varepsilon_1^2 \\
& = C \varepsilon_1^2 \\
& =: \sigma^2
\end{aligned}
\end{equation*}
where the third line follows by exchanging the trace operation and expectation, the fourth line uses Lemma \ref{lem:expvalunitsphere}, and the fifth line follows from Lemma \ref{lem:frobeniusnorm2}.

The matrix Bernstein inequality and the union bound yield
\begin{equation*}
\begin{aligned}
\mathbb{P} \left\{ \sup_{\tau_d \in \Omega_{\Grid}}\norm[2]{\bI_2^{\ell}(\tau_d)} \geq \varepsilon_5 \Bigg| \cE_{1, \varepsilon_1} \right\} 
& \leq |\Omega_{\Grid}|
\mathbb{P} \left\{ \| \sum_{j} \widetilde{\bZ}_j\|  \geq \varepsilon_5 \Bigg|  \cE_{1, \varepsilon_1} \right\} \\
& \leq |\Omega_{\Grid}| (K + 1) \cdot \exp\left(\frac{-\varepsilon_5^2/2}{\sigma^2 + R\varepsilon_5 /3}\right) \\
& \leq
{\begin{cases}
|\Omega_{\Grid}|(K + 1) \exp \left(\frac{-3\varepsilon_5^2}{8\sigma^2} \right), & \varepsilon_5 \leq \sigma^2 /R \\
|\Omega_{\Grid}|(K + 1) \exp \left(\frac{-3\varepsilon_5}{8R} \right), & \varepsilon_5 \geq \sigma^2/R.
\end{cases}}
\end{aligned}
\end{equation*}
Therefore, we can write
\begin{equation*}
\begin{aligned}
& \mathbb{P} \left\{ \sup_{\tau_d \in \Omega_{\Grid}}\norm[2]{\bI_2^{\ell}(\tau_d)}\geq \varepsilon_5 \right\} \\
&  \leq {\begin{cases}
|\Omega_{\Grid}|(K + 1) \exp \left(\frac{-3\varepsilon_5^2}{8\sigma^2} \right) + \mathbb{P}\left(\cE_{1, \varepsilon_1}^c\right), & \varepsilon_5 \leq \sigma^2 /R \\
|\Omega_{\Grid}|(K + 1) \exp \left(\frac{-3\varepsilon_5}{8R} \right)+ \mathbb{P}\left(\cE_{1, \varepsilon_1}^c\right), & \varepsilon_5 \geq \sigma^2/R.
\end{cases}}
\end{aligned}
\end{equation*}
When $\varepsilon_5 \leq \sigma^2 /R$, to ensure the first term less than $\delta$, it is sufficient to choose
\begin{equation*}
\varepsilon_1^2 = \frac{3\varepsilon_5^2} {8C \log\left(\frac{|\Omega_{\Grid}|(K+1)}{\delta}\right)}
\end{equation*}
for the same constant $C$ that appears in the variance bound above.

To make the second term $\mathbb{P}\left(\cE_{1, \varepsilon_1}^c\right)$ less than $\delta$, according to Lemma \ref{lem:matrixbernstein}, we have
\begin{equation*}
\begin{aligned}
M & \geq \frac{80 \mu JK}{\varepsilon_1^2} \log\left( \frac{4JK}{\delta}\right) \\
& = \frac{8 C \cdot 80 \mu JK}{3\varepsilon_5^2} \log\left( \frac{4JK}{\delta}\right)\log\left(\frac{|\Omega_{\Grid}|(K+1)}{\delta}\right) \\
& =: C \frac{\mu JK}{\varepsilon_5^2}\log \left(\frac{JK}{\delta}\right)\log\left(\frac{|\Omega_{\Grid}|K}{\delta}\right)
\end{aligned}
\end{equation*}
with a redefined numerical constant $C$.

Similarly, when $\varepsilon_5 \geq \sigma^2 /R$, to ensure the first term less than $\delta$, we can take
\begin{equation*}
\varepsilon_1 = \frac{3\varepsilon_5}{8C \log\left(\frac{|\Omega_{\Grid}|(K+1)}{\delta}\right)},
\end{equation*}
for the same constant $C$ shown in the bound of $\norm[]{\widetilde{\bZ}_j}$.

To make the second term less than $\delta$, we have
\begin{equation*}
\begin{aligned}
M & \geq \frac{80 \mu JK}{\varepsilon_1^2} \log\left( \frac{4JK}{\delta}\right) \\
& = \frac{(8 C)^2 \cdot 80 \mu JK}{9\varepsilon_5^2} \log\left( \frac{4JK}{\delta}\right)\log^2\left(\frac{|\Omega_{\Grid}|(K+1)}{\delta}\right) \\
& =: C \frac{\mu JK}{\varepsilon_5^2}\log \left(\frac{JK}{\delta}\right)\log^2\left(\frac{|\Omega_{\Grid}|K}{\delta}\right)
\end{aligned}
\end{equation*}
for a redefined numerical constant $C$.

Combining the two different cases above and applying the union bound with respect to $\ell = 0,1,2,3$ complete the proof.
\section{Proof of Lemma \ref{lem:continuousextension}}
\label{proofofLemma5_11}
Denote the $p$th column of $\bV_{\ell}(\tau)$ as $\bV_{\ell}(\tau;p)$, whose $\ell_2$-norm can be bounded as follows:
\begin{equation*}
\begin{aligned}
\norm[2]{\bV_{\ell}(\tau;p)} &  = \norm[2]{\frac{1}{M} \sum_{n = -2M}^{2M} g_M(n)\left(\frac{-i2\pi n}{\sqrt{|K''_M(0)|}}\right)^{\ell}\vphantom{\frac{1}{M} \sum_{n = -2M}^{2M} g_M(n)\left(\frac{-i2\pi n}{\sqrt{|K''_M(0)|}}\right)^{\ell} } e^{-i2\pi n \tau} \bE(n) \otimes \bb_n \bb_n^{*}(p)} \\
& \leq \frac{(4M+1)4^{\ell}}{M} \mu \sqrt{ K} \norm[2]{\bE(n)} \\
& \leq \frac{\sqrt{14}\cdot (4M+1)4^{\ell}}{M} \mu \sqrt{ JK} \\
& = C \mu \sqrt{ JK}
\end{aligned}
\end{equation*}
for some constant $C$,
where we have used the fact that $|\bb_n(p)|\leq \sqrt{\mu}$ in the first inequality and the fact that $\norm[2]{\bE(n)}^2 \leq 14J$ when $M\geq 4$ in the second inequality.

We define the $p$th entry of $\bq^{\ell}(\tau)$ as $\bq^{\ell}(\tau;p)$. Conditioned on the event $\cE_{1, \varepsilon_1}$ with $\varepsilon_1\in (0,\frac{1}{4}]$, we have
\begin{equation*}
\begin{aligned}
\left| \frac{1}{\sqrt{|K''_M(0)|}^{\ell}} \bq^{\ell} (\tau;p) \right|& = \left|\bV_{\ell}(\tau; p)^H \bL \bh\right| \\
& \leq \norm[2]{\bV_{\ell}(\tau;p)} \norm[]{\bL} \norm[2]{\bh} \\
& \leq C \mu  J\sqrt{K}
\end{aligned}
\end{equation*}
for some constant $C$.

Applying Bernstein's polynomial inequality  \cite{schaeffer1941inequalities,tang2013compressed}, we have
\begin{equation*}
\begin{aligned}
& \left| \frac{1}{\sqrt{|K''_M(0)|}^{\ell}} \bq^{\ell} (\tau_a;p) -  \frac{1}{\sqrt{|K''_M(0)|}^{\ell}} \bq^{\ell} (\tau_b;p)\right|\\
& \leq |e^{i2\pi \tau_a} - e^{i2\pi \tau_b} |\sup_{z = e^{i2\pi \tau}}\left|\frac{d\frac{1}{\sqrt{|K''_M(0)|}^{\ell}} \bq^{\ell} (z;p)}{dz}\right| \\
& \leq 4\pi |\tau_a -\tau_b| 2M \sup_{\tau}\left|\frac{1}{\sqrt{|K''_M(0)|}^{\ell}} \bq^{\ell} (\tau;p)\right| \\
& \leq C M \mu J\sqrt{ K}  |\tau_a -\tau_b|
\end{aligned}
\end{equation*}
for some numerical constant $C$.
Therefore, we have
\begin{equation*}
\begin{aligned}
\norm[2]{\frac{1}{\sqrt{|K''_M(0)|}^{\ell}} \bq^{\ell} (\tau_a) -  \frac{1}{\sqrt{|K''_M(0)|}^{\ell}} \bq^{\ell} (\tau_b)} 
& \leq C \mu JK M|\tau_a-\tau_b| \\
& \leq CM^2 |\tau_a-\tau_b|,
\end{aligned}
\end{equation*}
where the second line follows when $M\geq \mu JK$. We choose $\Omega_{\Grid}$ such that for any $\tau \in [0,1)$, there exists a point $\tau_d \in \Omega_{\Grid}$ with $|\tau-\tau_d|\leq \frac{\varepsilon}{3CM^2}$. Note that $|\Omega_{\Grid}| \leq \frac{3CM^2}{\varepsilon}$.

Using such a choice of $\Omega_{\Grid}$ and conditioned on the event $\cE_{1,\varepsilon_1}$ with $\varepsilon_1\in (0,\frac{1}{4}]$ and event $\cE$, we have \begin{equation*}
\begin{aligned}
& \norm[2]{\frac{1}{\sqrt{|K''_M(0)|}^{\ell}} \bq^{\ell} (\tau) -   \frac{1}{\sqrt{|K''_M(0)|}^{\ell}} \overline{\bq}^{\ell} (\tau)} \\
& \leq \norm[2]{\frac{1}{\sqrt{|K''_M(0)|}^{\ell}} \bq^{\ell} (\tau) -   \frac{1}{\sqrt{|K''_M(0)|}^{\ell}} \bq^{\ell} (\tau_d)} \\
& ~~~~+ \norm[2]{\frac{1}{\sqrt{|K''_M(0)|}^{\ell}} \bq^{\ell} (\tau_d) -   \frac{1}{\sqrt{|K''_M(0)|}^{\ell}} \overline{\bq}^{\ell} (\tau_d)} \\
& ~~~~+  \norm[2]{ \frac{1}{\sqrt{|K''_M(0)|}^{\ell}} \overline{\bq}^{\ell} (\tau_d) -   \frac{1}{\sqrt{|K''_M(0)|}^{\ell}} \overline{\bq}^{\ell} (\tau)}\\
& \leq CM^2 |\tau-\tau_d| + \frac{\varepsilon}{3} + CM^2 |\tau-\tau_d| \\
& \leq \varepsilon, ~~\forall~\tau \in [0,1)
\end{aligned}
\end{equation*}
where the first inequality follows from the triangle inequality, the second inequality follows from Proposition \ref{prop:concentrationeq}. With such a choice of grid size and by applying Proposition \ref{prop:concentrationeq}, we can immediately get the following bound on $M$:
\begin{equation*}
\begin{aligned}
M \geq C\mu JK \max \left\{ \frac{1}{\varepsilon^2}\log\left(\frac{MJK}{\varepsilon \delta}\right) \log^2\left(\frac{MK}{\varepsilon \delta}\right),\log\left(\frac{JK}{\delta} \right)\right\}
\end{aligned}
\end{equation*}
with a redefined numerical constant $C$.
This finishes the proof of Lemma \ref{lem:continuousextension}.

\section{Supplementary Materials for Lemma \ref{lem:nearregionproof}}
\label{sec:suppmaterial}
First of all, we record some useful results from the proof of Lemmas 2.3 and 2.4 in \cite{candes2014towards}.

Assume that $\Delta_{\tau}\geq \frac{1}{M}, M\geq 64$. Then, we have
\begin{equation*}
0.9539 \leq K(\tau) \leq 1, \tau \in [-\tau_{b,1}, \tau_{b,1}],
\end{equation*}
\begin{equation*}
-13.572 M^2 \leq K''(\tau) \leq -11.692 M^2, \tau \in [-\tau_{b,1}, \tau_{b,1}],
\end{equation*}
\begin{equation*}
\sum_{\tau_j \in \mathbb{D}} |K'(\tau-\tau_j)| \leq 1.2722M,\tau \in [-\tau_{b,1}, \tau_{b,1}],
\end{equation*}
\begin{equation*}
\sum_{\tau_j \in \mathbb{D}} |K'''(\tau-\tau_j)| \leq  194.0560 M^3,\tau \in [-\tau_{b,1}, \tau_{b,1}],
\end{equation*}
\begin{equation*}
\sum_{\tau_j \in \mathbb{D}\setminus 0 } |K(\tau-\tau_j)| \leq 6.279\times 10^{-3},\tau \in \Omega_{\near} \setminus [-\tau_{b,1}, \tau_{b,1}],
\end{equation*}
\begin{equation*}
\sum_{\tau_j \in \mathbb{D}\setminus 0 } |K''(\tau-\tau_j)| \leq 4.2200M^2,\tau \in \Omega_{\near} \setminus [-\tau_{b,1}, \tau_{b,1}].
\end{equation*}

We know that
\begin{equation*}
\overline{\bq}(\tau) = \sum_{j = 1}^J K_M(\tau-\tau_j) \overline{\balpha}_j +  \sum_{j =1}^J K_{M}'(\tau-\tau_j) \overline{\bbeta}_j.
\end{equation*}
It was shown in \cite{yang2014exact} that
\begin{equation*}
\begin{aligned}
1-8.824\times 10^{-3} & = \alpha_{\min}
 \leq \norm[2]{\overline{\balpha}_j} \leq \alpha_{\max} = 1+ 8.824\times 10^{-3}
\end{aligned}
\end{equation*}
and
\begin{equation*}
\norm[2]{\overline{\bbeta}_j}\leq \beta_{\max} = \frac{1.647}{M} \times 10^{-2}.
\end{equation*}
Without loss of generality, in the following, we assume that $0\in \mathbb{D}$. Thus, we have
\begin{equation*}
\begin{aligned}
\norm[2]{\overline{ \bq}(\tau)} 
& =  \norm[2]{\sum_{\tau_j \in \mathbb{D}} K(\tau-\tau_j)\overline{\balpha}_j + \sum_{\tau_j \in \mathbb{D}}K'(\tau-\tau_j) \overline{\bbeta}_j} \\
& \leq  \alpha_{\max} \sum_{\tau_j \in \mathbb{D}} |K(\tau-\tau_j)| + \beta_{\max}
\sum_{\tau_j \in \mathbb{D}} |K'(\tau-\tau_j)| \\
& \leq 1.008824\times (1+ 6.279\times 10^{-3}) + \frac{1.647\times 10^{-2}}{M} \times (1.2722M)\\
& = 1.0361,
\end{aligned}
\end{equation*}
\begin{equation*}
\begin{aligned}
\norm[2]{\overline{ \bq}'(\tau)} 
 & =  \norm[2]{\sum_{\tau_j \in \mathbb{D}} K'(\tau-\tau_j)\overline{\balpha}_j + \sum_{\tau_j \in \mathbb{D}}K''(\tau-\tau_j) \overline{\bbeta}_j} \\
& \leq \alpha_{\max} \sum_{\tau_j \in \mathbb{D}} |K'(\tau-\tau_j)| + \beta_{\max} |K''(\tau)| + \beta_{\max} \sum_{\tau_j \in \mathbb{D}\setminus 0}|K''(\tau-\tau_j)| \\
& \leq 1.008824 \times (1.2722M) + \frac{1.647 \times 10^{-2}}{M}\times (13.572 M^2 + 4.2200M^2) \\
& = 1.5765M,
\end{aligned}
\end{equation*}

\begin{equation*}
\begin{aligned}
\norm[2]{\overline{\bq}''(\tau)} 
& =  \norm[2]{\sum_{\tau_j \in \mathbb{D}} K''(\tau-\tau_j)\overline{\balpha}_j + \sum_{\tau_j \in \mathbb{D}}K'''(\tau-\tau_j) \overline{\bbeta}_j} \\
& \leq\alpha_{\max} |K''(\tau)| +  \alpha_{\max} \sum_{\tau_j \in \mathbb{D}\setminus 0} |K''(\tau-\tau_j)| + \beta_{\max} \sum_{\tau_j \in \mathbb{D}}|K'''(\tau-\tau_j)| \\
& \leq 1.008824 \times(13.572M^2 + 4.2200M^2) + \frac{1.647\times 10^{-2}}{M}\times (194.0560M^3) \\
& = 21.1451M^2,
\end{aligned}
\end{equation*}

\begin{equation*}
\begin{aligned}
& \overline{\bq}''^H(\tau) \overline{ \bq}(\tau) \\
& =  \left(\sum_{\tau_j \in \mathbb{D}} K''(\tau-\tau_j)\overline{\balpha}_j + \sum_{\tau_j \in \mathbb{D}}K'''(\tau-\tau_j) \overline{\bbeta}_j\right)^H\\
&~~~~\times \left( \sum_{\tau_j \in \mathbb{D}} K(\tau-\tau_j)\overline{\balpha}_j + \sum_{\tau_j \in \mathbb{D}}K'(\tau-\tau_j) \overline{\bbeta}_j\right) \\
&  = \norm[2]{\overline{\balpha}_k}^2 K''(\tau)K(\tau) + \overline{\balpha}_k^H K''(\tau)  \sum_{\tau_j \in \mathbb{D}\setminus 0}K(\tau-\tau_j) \overline{\balpha}_j \\
& ~~~~+ \overline{\balpha}_k^H K''(\tau)  \sum_{\tau_j \in \mathbb{D}}K'(\tau-\tau_j) \overline{\bbeta}_j +  \left(\sum_{\tau_j \in \mathbb{D}\setminus 0}K''(\tau-\tau_j) \overline{\balpha}_j\right)^H \overline{\bq}(\tau) \\
& ~~~~+\left(\sum_{\tau_j \in \mathbb{D}}K'''(\tau-\tau_j) \overline{\bbeta}_j\right)^H \overline{\bq}(\tau).
\end{aligned}
\end{equation*}
Now we upper bound the following terms:
\begin{equation*}
\begin{aligned}
\norm[2]{\overline{\balpha}_k}^2 K''(\tau)K(\tau) 
& \leq \alpha_{\min}^2 \times(-11.692M^2) \times 0.9539\\
& = -10.9570M^2,
\end{aligned}
\end{equation*}
\begin{equation*}
\begin{aligned}
{\text{Re}}\left(\overline{\balpha}_k^H K''(\tau)  \sum_{\tau_j \in \mathbb{D}\setminus 0}K(\tau-\tau_j) \overline{\balpha}_j\right) & \leq \alpha_{\max}^2 |K''(\tau)| \sum_{\tau_j \in \mathbb{D}\setminus 0}|K(\tau-\tau_j)| \\
& \leq 1.008824^2 \times (13.572M^2) \times (6.279\times10^{-3}) \\
& = 0.0867M^2,
\end{aligned}
\end{equation*}

\begin{equation*}
\begin{aligned}
{\text{Re}}\left(\overline{\balpha}_k^H K''(\tau)  \sum_{\tau_j \in \mathbb{D}}K'(\tau-\tau_j) \overline{\bbeta}_j \right) 
& \leq \alpha_{\max} \beta_{\max} |K''(\tau)| \sum_{\tau_j \in \mathbb{D}}|K'(\tau-\tau_j)| \\
& \leq 1.008824\times \frac{1.647\times 10^{-2}}{M} \times (13.572M^2) \times (1.2722M)\\
& = 0.2869M^2,
\end{aligned}
\end{equation*}

\begin{equation*}
\begin{aligned}
{\text{Re}}\left(\left(\sum_{\tau_j \in \mathbb{D}\setminus 0}K''(\tau-\tau_j) \overline{\balpha}_j\right)^H \overline{\bq}(\tau)\right) 
& \leq \alpha_{\max} \sum_{\tau_j \in \mathbb{D}\setminus 0}|K''(\tau-\tau_j)| \norm[2]{\overline{\bq}(\tau)} \\
& \leq 1.008824\times (4.2200M^2) \times 1.0361\\
& =  4.4109M^2,
\end{aligned}
\end{equation*}

\begin{equation*}
\begin{aligned}
{\text{Re}}\left( \left(\sum_{\tau_j \in \mathbb{D}}K'''(\tau-\tau_j) \overline{\bbeta}_j\right)^H \overline{\bq}(\tau)\right)
 & \leq \beta_{\max} \sum_{\tau_j \in \mathbb{D}}|K'''(\tau-\tau_j) |\norm[2]{\overline{\bq}(\tau)} \\
& \leq \frac{1.647\times 10^{-2}}{M} \times (194.0560M^3) \times 1.0361\\
& = 3.3115M^2.
\end{aligned}
\end{equation*}
Combining the above upper bounds, we have
\begin{equation*}
\begin{aligned}
{\text{Re}}(\overline{\bq}''^H(\tau) \overline{\bq}(\tau))
  & \leq (-10.9570+0.0867+0.2869+4.4109+3.3115)M^2\\
& = -2.8610M^2.
\end{aligned}
\end{equation*}
Consequently, we have
\begin{equation*}
\begin{aligned}
\norm[2]{\overline{\bq}'(\tau)}^2 + {\text{Re}}(\overline{\bq}''^H(\tau) \overline{\bq}(\tau))
& \leq (1.5765M)^2  -2.8610M^2 \\
& = -0.3756M^2.
\end{aligned}
\end{equation*}

\bibliographystyle{IEEEbib}
\footnotesize
\bibliography{refs}

\begin{thebibliography}{10}

\bibitem{yang2016non}
D.~Yang, G.~Tang, and M.~Wakin,
\newblock ``Non-stationary blind super-resolution,''
\newblock in {\em Proceedings of IEEE International Conference on Acoustics,
  Speech and Signal Processing (ICASSP)}, Shanghai, China, March 2016, pp.
  4727--4731.

\bibitem{huang2008three}
B.~Huang, W.~Wang, M.~Bates, and X.~Zhuang,
\newblock ``Three-dimensional super-resolution imaging by stochastic optical
  reconstruction microscopy,''
\newblock {\em Science}, vol. 319, no. 5864, pp. 810--813, 2008.

\bibitem{yildiz2003myosin}
A.~Yildiz, J.~Forkey, S.~McKinney, T.~Ha, Y.~Goldman, and P.~Selvin,
\newblock ``Myosin {V} walks hand-over-hand: Single fluorophore imaging with
  1.5-nm localization,''
\newblock {\em Science}, vol. 300, no. 5628, pp. 2061--2065, 2003.

\bibitem{margrave2011gabor}
G.~Margrave, M.~Lamoureux, and D.~Henley,
\newblock ``Gabor deconvolution: Estimating reflectivity by nonstationary
  deconvolution of seismic data,''
\newblock {\em Geophysics}, vol. 76, no. 3, pp. W15--W30, 2011.

\bibitem{quirin2012optimal}
S.~Quirin, S.~Pavani, and R.~Piestun,
\newblock ``Optimal 3{D} single-molecule localization for superresolution
  microscopy with aberrations and engineered point spread functions,''
\newblock {\em Proceedings of the National Academy of Sciences}, vol. 109, no.
  3, pp. 675--679, 2012.

\bibitem{fergus2006removing}
R.~Fergus, B.~Singh, A.~Hertzmann, S.~Roweis, and W.~Freeman,
\newblock ``Removing camera shake from a single photograph,''
\newblock {\em ACM Trans. Graphics}, vol. 25, no. 3, pp. 787--794, 2006.

\bibitem{starck2002deconvolution}
J.~Starck, E.~Pantin, and F.~Murtagh,
\newblock ``Deconvolution in astronomy: A review,''
\newblock {\em Publications of the Astronomical Society of the Pacific}, vol.
  114, no. 800, pp. 1051--1069, 2002.

\bibitem{luo2006low}
X.~Luo and G.~Giannakis,
\newblock ``Low-complexity blind synchronization and demodulation for (ultra-)
  wideband multi-user ad hoc access,''
\newblock {\em IEEE Trans. Wireless Communications}, vol. 5, no. 7, pp.
  1930--1941, 2006.

\bibitem{heckel2014super}
R.~Heckel, V.~Morgenshtern, and M.~Soltanolkotabi,
\newblock ``Super-resolution radar,''
\newblock {\em Information and Inference: A Journal of the IMA}, 2016, to
  appear.

\bibitem{cai2015robust}
J.~Cai, X.~Qu, W.~Xu, and G.~Ye,
\newblock ``Robust recovery of complex exponential signals from random
  {G}aussian projections via low-rank {H}ankel matrix reconstruction,''
\newblock {\em Applied and Computational Harmonic Analysis}, 2016, to appear.

\bibitem{candes2014towards}
E.~Cand{\`e}s and C.~Fernandez-Granda,
\newblock ``Towards a mathematical theory of super-resolution,''
\newblock {\em Communications on Pure and Applied Mathematics}, vol. 67, no. 6,
  pp. 906--956, 2014.

\bibitem{bhaskar2013atomic}
B.~Bhaskar, G.~Tang, and B.~Recht,
\newblock ``Atomic norm denoising with applications to line spectral
  estimation,''
\newblock {\em IEEE Trans. Signal Processing}, vol. 61, no. 23, pp. 5987--5999,
  2013.

\bibitem{candes2013super}
E.~Cand{\`e}s and C.~Fernandez-Granda,
\newblock ``Super-resolution from noisy data,''
\newblock {\em Journal of Fourier Analysis and Applications}, vol. 19, no. 6,
  pp. 1229--1254, 2013.

\bibitem{tang2015near}
G.~Tang, B.~Bhaskar, and B.~Recht,
\newblock ``Near minimax line spectral estimation,''
\newblock {\em IEEE Trans. Information Theory}, vol. 61, no. 1, pp. 499--512,
  2015.

\bibitem{morgenshtern2015super}
V.~Morgenshtern and E.~Cand{\`e}s,
\newblock ``Super-resolution of positive sources: The discrete setup,''
\newblock {\em SIAM Journal on Imaging Sciences}, vol. 9, no. 1, pp. 412--444,
  2016.

\bibitem{duval2015exact}
V.~Duval and G.~Peyr{\'e},
\newblock ``Exact support recovery for sparse spikes deconvolution,''
\newblock {\em Foundations of Computational Mathematics}, vol. 15, no. 5, pp.
  1315--1355, 2015.

\bibitem{schiebinger2015superresolution}
G.~Schiebinger, E.~Robeva, and B.~Recht,
\newblock ``Superresolution without separation,''
\newblock {\em arXiv preprint arXiv:1506.03144}, 2015.

\bibitem{fernandez2015super}
C.~Fernandez-Granda,
\newblock ``Super-resolution of point sources via convex programming,''
\newblock {\em Information and Inference: A Journal of the IMA}, 2016, to
  appear.

\bibitem{tang2013compressed}
G.~Tang, B.~Bhaskar, P.~Shah, and B.~Recht,
\newblock ``Compressed sensing off the grid,''
\newblock {\em IEEE Trans. Information Theory}, vol. 59, no. 11, pp.
  7465--7490, 2013.

\bibitem{li2014off}
Y.~Li and Y.~Chi,
\newblock ``Off-the-grid line spectrum denoising and estimation with multiple
  measurement vectors,''
\newblock {\em IEEE Trans. Signal Processing}, vol. 64, no. 5, pp. 1257--1269,
  2016.

\bibitem{yang2014exact}
Z.~Yang and L.~Xie,
\newblock ``Exact joint sparse frequency recovery via optimization methods,''
\newblock {\em arXiv preprint arXiv:1405.6585}, 2014.

\bibitem{chen2014robust}
Y.~Chen and Y.~Chi,
\newblock ``Robust spectral compressed sensing via structured matrix
  completion,''
\newblock {\em IEEE Trans. Information Theory}, vol. 60, no. 10, pp.
  6576--6601, 2014.

\bibitem{ahmed2014blind}
A.~Ahmed, B.~Recht, and J.~Romberg,
\newblock ``Blind deconvolution using convex programming,''
\newblock {\em IEEE Trans. Information Theory}, vol. 60, no. 3, pp. 1711--1732,
  2014.

\bibitem{aliahmedcov}
A.~Ahmed, A.~Cosse, and L.~Demanet,
\newblock ``A convex approach to blind deconvolution with diverse input,''
\newblock {\em Preprint}, 2015.

\bibitem{ling2015blind}
S.~Ling and T.~Strohmer,
\newblock ``Blind deconvolution meets blind demixing: {A}lgorithms and
  performance bounds,''
\newblock {\em arXiv preprint arXiv:1512.07730}, 2015.

\bibitem{lee2013near}
K.~Lee, Y.~Wu, and Y.~Bresler,
\newblock ``Near optimal compressed sensing of sparse rank-one matrices via
  sparse power factorization,''
\newblock {\em arXiv preprint arXiv:1312.0525}, 2013.

\bibitem{tang2014convex}
G.~Tang and B.~Recht,
\newblock ``Convex blind deconvolution with random masks,''
\newblock in {\em Computational Optical Sensing and Imaging (COSI)}, 2014.

\bibitem{bahmani2014lifting}
S.~Bahmani and J.~Romberg,
\newblock ``Lifting for blind deconvolution in random mask imaging:
  Identifiability and convex relaxation,''
\newblock {\em SIAM Journal on Imaging Sciences}, vol. 8, no. 4, pp.
  2203--2238, 2015.

\bibitem{li2015identifiability}
Y.~Li, K.~Lee, and Y.~Bresler,
\newblock ``Identifiability in blind deconvolution with subspace or sparsity
  constraints,''
\newblock {\em IEEE Trans. Information Theory}, vol. 62, no. 7, pp. 4266--4275,
  2016.

\bibitem{choudhary2014identifiability}
S.~Choudhary and U.~Mitra,
\newblock ``Identifiability scaling laws in bilinear inverse problems,''
\newblock {\em arXiv preprint arXiv:1402.2637}, 2014.

\bibitem{chi2015guaranteed}
Y.~Chi,
\newblock ``Guaranteed blind sparse spikes deconvolution via lifting and convex
  optimization,''
\newblock {\em IEEE Journal of Selected Topics in Signal Processing}, vol. 10,
  no. 4, pp. 782--794, 2016.

\bibitem{ling2015self}
S.~Ling and T.~Strohmer,
\newblock ``Self-calibration and biconvex compressive sensing,''
\newblock {\em Inverse Problems}, vol. 31, no. 11, pp. 115002, 2015.

\bibitem{oymak2015simultaneously}
S.~Oymak, A.~Jalali, M.~Fazel, Y.~Eldar, and B.~Hassibi,
\newblock ``Simultaneously structured models with application to sparse and
  low-rank matrices,''
\newblock {\em IEEE Trans. Information Theory}, vol. 61, no. 5, pp. 2886--2908,
  2015.

\bibitem{candes2015phase}
E.~Cand{\`e}s, Y.~Eldar, T.~Strohmer, and V.~Voroninski,
\newblock ``Phase retrieval via matrix completion,''
\newblock {\em SIAM Review}, vol. 57, no. 2, pp. 225--251, 2015.

\bibitem{chandrasekaran2012convex}
V.~Chandrasekaran, B.~Recht, P.~Parrilo, and A.~Willsky,
\newblock ``The convex geometry of linear inverse problems,''
\newblock {\em Foundations of Computational Mathematics}, vol. 12, no. 6, pp.
  805--849, 2012.

\bibitem{grant2008cvx}
M.~Grant, S.~Boyd, and Y.~Ye,
\newblock ``C{VX}: Matlab software for disciplined convex programming,'' 2008.

\bibitem{candes2011probabilistic}
E.~Cand{\`e}s and Y.~Plan,
\newblock ``A probabilistic and {RIP}less theory of compressed sensing,''
\newblock {\em IEEE Trans. Information Theory}, vol. 57, no. 11, pp.
  7235--7254, 2011.

\bibitem{tropp2012user}
J.~Tropp,
\newblock ``User-friendly tail bounds for sums of random matrices,''
\newblock {\em Foundations of Computational Mathematics}, vol. 12, no. 4, pp.
  389--434, 2012.

\bibitem{schaeffer1941inequalities}
A.~Schaeffer,
\newblock ``Inequalities of {A}. {M}arkoff and {S}. {B}ernstein for polynomials
  and related functions,''
\newblock {\em Bulletin of the American Mathematical Society}, vol. 47, no. 8,
  pp. 565--579, 1941.

\end{thebibliography}

\end{document}